\documentclass[final,leqno]{siamltex}

\usepackage{mathtools}
\usepackage{amsmath}	
\usepackage{amssymb}	
\usepackage{graphicx,graphics,color}	
\usepackage{enumerate}
\usepackage{cancel}
\usepackage{color}

\usepackage{algorithmic}	
\usepackage[section]{algorithm}

\newcommand{\cov}{\operatornamewithlimits{Cov}}
\newcommand{\vect}{\operatornamewithlimits{vec}}
\newcommand{\E}{\mbox{E}}

\def\ci{\perp\!\!\!\perp}

\newtheorem{thm}{Theorem}[section] 
\newtheorem{remark}[thm]{Remark}

\title{Causal Network Inference \\ by Optimal Causation Entropy\thanks{This work was funded by ARO Grant No. 61386-EG (J.S and E.M.B), and NSF Grant No. DMS-1127914 through the Statistical and Applied Mathematical Sciences Institute (D.T.).}}

\author{Jie Sun\thanks{Department of Mathematics, Clarkson University, Potsdam, NY 13699 ({\tt sunj@clarkson.edu}).}
        \and Dane Taylor\thanks{Statistical and Applied Mathematical Sciences Institute, Research Triangle Park, NC 27709; Department of Mathematics, University of North Carolina, Chapel Hill, NC 27599.}
        \and Erik M. Bollt\thanks{Department of Mathematics, Clarkson University, Potsdam, NY 13699.}}

\begin{document}

\maketitle

\begin{abstract}
The broad abundance of time series data, which is in sharp contrast to limited knowledge of the underlying network dynamic processes that produce such observations, calls for a rigorous and efficient method of causal network inference.
Here we develop mathematical theory of causation entropy, an information-theoretic statistic designed for model-free causality inference. For stationary Markov processes, we prove that for a given node in the network, its causal parents forms the {\it minimal set of nodes that maximizes causation entropy}, a result we refer to as the {\it optimal causation entropy principle}.
Furthermore, this principle guides us to develop computational and data efficient algorithms for causal network inference based on a two-step discovery and removal algorithm for time series data for a network-couple dynamical system.
Validation in terms of analytical and numerical results for Gaussian processes on large random networks highlight that inference by our algorithm outperforms previous leading methods including conditioned Granger causality and transfer entropy.
Interestingly, our numerical results suggest that the number of samples required for accurate inference depends strongly on network characteristics such as the density of links and information diffusion rate and not necessarily on the number of nodes.
\end{abstract}



\pagestyle{myheadings}
\thispagestyle{plain}
\markboth{J. SUN, D. TAYLOR, AND E. M. BOLLT}{CAUSAL NETWORK INFERENCE BY OPTIMAL CAUSATION ENTROPY}

\section{Introduction}
Research of dynamic processes on large-scale complex networks has attracted considerable interest in recent years with exciting developments in a wide range of disciplines in social, scientific, engineering, and medical fields~\cite{Newman2003,NewmanBook,Traud2011}. One important line of research focuses on exploring the role of network structure in determining the dynamic properties of a system~\cite{BarratBook,Craciun2010,Dorfler2012,Dorogovtsev2008,Golubitsky2005,Pomerance2009,Stilwell2006,Watts2000} and utilizing such knowledge in controlling network dynamics~\cite{Cornelius2013,Sun2013PRL} and optimizing network performance~\cite{Chen2009,Kleinberg2000,Nishikawa2010,Ravoori2011,Taylor2011}.
In applications such as the study of neuronal connectivity or gene interactions, it is nearly impossible to directly identify the network structure without severely interfering with the underlying system whereas time series measurements of the individual node states are often more accessible~\cite{Stolovitzky2007}. From this perspective, it is crucial to reliably infer the network structure that shapes the dynamics of a system from time series data. It is essential that one accounts for directed ``cause and effect'' relationships, which often offer deeper insight than non-directed relationships (e.g., correlations)~\cite{Pearl2009,Schindlera2007,Spirtes2000}. 
In particular, causal network inference is considered a central problem in the research of social perception~\cite{Heider1944}, epidemiological factors~\cite{Rothman2005}, neural connectivity~\cite{Bassett2006,Bullmore2009}, economic impacts~\cite{Heckman2008}, and basic physical relationships of climatological events~\cite{Runge2012PRL,Runge2012PRE}. 
Evidently, understanding causality is a necessary and important precursor step towards the goal of effectively controlling and optimizing system dynamics (e.g., medical intervention of biological processes and policy design for economic growth and social development).

In a network dynamic process involving a large number of nodes, causal relationships are inherently difficult to infer. 
For example, the fact that a single node can potentially be influenced by many (if not all) others through network interactions makes it challenging to untangle the direct causal links from indirect and erroneous ones (see Fig.~\ref{newfig:1} for illustration).
Granger recognized the crucial role played by conditioning and defines a causal relationship based on two basic principles~\cite{Granger1969,Granger1988}:
\begin{enumerate}[(i)]
\item The cause should occur before the effect; 
\item The cause should contain information about the caused that is not available otherwise.
\end{enumerate}
A relationship that fulfills both requirements is unambiguously defined as causal.
In practice, although the first requirement is straightforward to examine when temporal ordering of the data is available, it is difficult to check the second as it involves the consideration of {\it all} available information (time series data from all variables).
Tradeoffs are often made, by either restricting to small-scale networks with no time delay and just a handful of variables~\cite{Guo2008,Sun2013PhysicaD}, or partially removing the second requirement therefore reducing the accuracy of network inference~\cite{Vicente2011JCN}.
Inferring large-scale networks from time series data remains to be a relatively open problem~\cite{Kuchaiev2009,Stolovitzky2007}.

\begin{figure}[top]
\centering
\includegraphics*[width=1\textwidth]{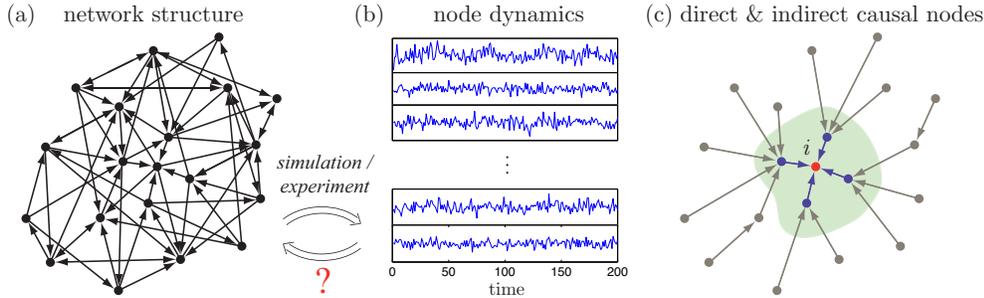}
\caption{
Network dynamics, time series, and the causal network inference problem.
Modern scientific approaches such as simulation, experiments, and data mining have produced an abundance of high-dimensional time series data describing dynamic processes on complex networks~(a$\to$b).
Given empirical observations, an important problem is to infer the causal network structure that underlies the observed time series. 
As shown in (c), for each node $i$, the goal is to identify its ``causal parents'', the nodes that directly influence its dynamics (nodes in shaded region), while pruning away the nodes that do not (nodes outside the shaded region), thus recovering the direct links to node $i$ in the causal network. 
The key to efficiently and accurately identify direct causal links from non-causal ones is to follow an algorithm involving tests for independence via judiciously selected conditioning sets. The main goal of this paper is to develop and validate such algorithms for stationary Markov processes.
}\label{newfig:1}
\end{figure}

The classical Granger causality test was designed for linear regression models~\cite{Granger1969,Granger1988}, although several extensions have been proposed to nonlinear models, including local linear approximations~\cite{Chen2004PLA} and partial functional expansion via radial basis functions~\cite{Ancona2004}.
Information-based causality inference measures represent a systematically way of overcoming the model-dependent limitation in the linear Granger causality test.
In particular, Schreiber proposed transfer entropy as a measure of information flow, or effective coupling, between two processes regardless of the actual functional relationship between them~\cite{Schreiber2000}. The transfer entropy from process $Y$ to $X$ measures the uncertainty reduction of the future states of $X$ as a result of knowing the past of $Y$ given that the past of $X$ is already known, and is essentially the mutual information between the future of $X$ and history of $Y$ {\em conditioning} on the history of $X$~\cite{Kaiser2002,Palus2001}. 
Because of its ability to associate {temporal and spatial} directionality with coupling, transfer entropy has quickly started to gain popularity in a broad range of disciplines including bioinformatics, neuroscience, climatology and others, as a tool to infer effective pairwise coupling that underlie complex dynamic processes~\cite{Bollt2012IJBC,Vicente2011JCN}. 
However, transfer entropy, which was introduced specifically for detecting the directionality of information flow between two processes, has  fundamental limitations when applied {in a multivariate setting,} to the inference of networks~\cite{Smirnof2013,Sun2013PhysicaD}. In particular, without proper conditioning, inference based on transfer entropy tends to produces systematic errors due to, for example, the effects of indirect influences and dominance of neighbors~\cite{Sun2013PhysicaD}. As shown in Fig.~\ref{newfig:1}(c), the main purpose of this work is to identify for each node $i$  its ``causal parents'' that directly influence node $i$, while not falsely inferring indirect (i.e., non-causal) nodes.

Proper conditioning can distinguish between direct and indirect causal relationships, and it is thus unsurprising that conditioning is widely adopted as a key ingredient in many network inference methods \cite{Frenzel2007,Guo2008,Kaiser2002,Palus2001,Runge2012PRL,Runge2012PRE,Smirnof2013,Spirtes2000,Sun2013PhysicaD}; however, even within such a general theme, the inference of networks requires a theoretically sound approach that is also algorithmically reliable and efficient. For example, one must develop a strategy for choosing which potential links to examine and which nodes to condition on. {Thus we note two essential steps in causal network inference: (1) adopting a statistic for the inference of a causal relationship, and (2) developing an algorithm that iteratively employs step (1) to learn the causal network. Whereas accuracy, tractability, and generality of the chosen statistic is often the priority for (1), various challenges arise regarding (2). In particular, these often include minimizing the computational cost by reducing the number of statistics that needs to be computed, as well as reducing the error incurred by finite-sized data by keeping the size of the conditioning sets (i.e., the dimension of the estimation problem) as small as possible. In general, the inaccuracy when estimating statistical measures from finite data grows rapidly with dimensionality, making the dimensionality of the problem a priority for any networks containing more than a couple nodes.}

One approach for network inference is to test each candidate causal link conditioned on all other variables~\cite{Guo2008}. That is, a direct link $j\to i$ is inferred if such a relationship remains effective when conditioning on all other variables in the system. Although intuitive and correct in theory, this method requires computing a statistic in a sample space as high dimensional as the entire system and therefore falls short when applied to a large networks.
The PC algorithm~\cite{Spirtes2000} overcomes this difficulty by repeated testing of the candidate causal link conditioned on subsets of the remaining variables~\cite{Runge2012PRL,Runge2012PRE}. To be more specific, a link $j\to i$ is disqualified as a candidate causal relationship if it is insignificant when conditioned on some subset of the nodes. The advantage of the PC algorithm is that it reduces the dimensionality of the sample space the test of independence to be proportional to the size of the conditioning set (which in some cases can be much smaller than the system size). However, unless the maximum degree of the nodes are known {\it a priori}, the algorithm in principle needs to be performed for combinations of subsets as the conditioning sets up to the size of the entire network. 
In this respect, regardless of the dimensionality of the sample space, the combinatorial search itself can be computationally infeasible for moderate to large networks. In practice, tradeoff needs to be made between an algorithm's computational cost and data efficiency (in terms of the estimation of the test statistic).

In this paper we develop theory of causation entropy---a type of conditional mutual information designed for causal network inference. In particular, we prove the optimal causation entropy principle for {Markov} processes: the set of nodes that directly cause a given node is the unique minimal set of nodes that maximizes causation entropy. This principle allows us to convert the problem of causality inference into the optimization of causation entropy. We further show that this optimization problem, which appears to be combinatorial, can in fact be solved by simple {greedy} algorithms, which are both computational efficient and data efficient. We {verify} the effectiveness of the proposed algorithms through analytical and numerical investigations of Gaussian processes on various network types including trees, loops, and random networks. Somewhat surprisingly, our results suggest that it is the density of links and information diffusion rate rather than the number of nodes in a network that determines the minimal sample size required for accurate inference.

\section{Stochastic Process and Causal Network Inference}
We begin by introducing a theoretical framework for inferring causal networks from high-dimensional time series. This framework is general in that it is applicable to both linear and non-linear systems with or without added noise.

Consider a network (graph) $\mathcal{G}=(\mathcal{V},\mathcal{E})$, with $\mathcal{V}=\{1,2,\dots,n\}$ being the set of nodes and
$\mathcal{E}\subset\mathcal{V}\times\mathcal{V}\times\mathbb{R}$ being the set of weighted links (or edges).
The {\it adjacency matrix} $A=[A_{ij}]_{n\times n}$ is defined as
\begin{equation}
{A_{ij}=
\begin{cases}
\mbox{weight of the link $j\rightarrow i$,} & \mbox{if $j\rightarrow i$ in the network};\\
0, & \mbox{otherwise}.
\end{cases}}
\end{equation}
We use $\chi_0(A)$ to denote the corresponding unweighted adjacency matrix defined entry-wise by $\chi_0(A)_{ij}=1$ iff $A_{ij}\not=0$ and $\chi_0(A)_{ij}=0$ iff $A_{ij}=0$.
We define the set of {\em causal parents} of $i$ as 
\begin{equation}\label{eq:causalparents}
	N_i=\{j|A_{ij}\neq 0\}=\{j|\chi_0(A)_{ij}= 1\}.
\end{equation}
For a subset of nodes $I\subset\mathcal{V}$, we similarly define its set of {causal parents} as
\begin{equation}\label{eq:causalparents2}
	N_I = \cup_{i\in I}N_i.
\end{equation}
We consider stochastic network dynamics in the following form (for each node $i$)
\begin{equation}\label{eq:generalnonlinear}
	X^{(i)}_{t} = f_i\big(A_{i1}X^{(1)}_{t-1},A_{i2}X^{(2)}_{t-1},\dots,A_{ij}X^{(j)}_{t-1},\dots,A_{in}X^{(n)}_{t-1},\xi^{(i)}_t\big)
\end{equation}
where $X^{(i)}_t\in\mathbb{R}^d$ is a random variable representing the state of node $i$ at time $t$,
$\xi^{(i)}_t\in\mathbb{R}^d$ is the random fluctuation on node $i$ at time $t$, and 
$f_i:\mathbb{R}^{d\times(n+1)}\rightarrow\mathbb{R}^d$ models the functional dependence of the state of node $i$ on the past states of nodes $j$ with $A_{ij}\neq 0$.
Note that other than the noise term $\xi^{(i)}_t$, the state $X^{(i)}_t$ only depends (stochastically) on the past states of its {causal parents}, $X^{(j)}_{t-1}$ ($j\in N_i$).

For a subset $K=\{k_1,k_2,\dots,k_q\}\subset\mathcal{V}$, we define 
\begin{equation}
	X^{(K)}_t\equiv[X^{(k_1)}_t,X^{(k_2)}_t,\dots,X^{(k_q)}_t]^\top.
\end{equation}
If $K=\mathcal{V}$, we simplify the notation and denote
\begin{equation}
	X_t\equiv X^{(\mathcal{V})}_t=[X^{(1)}_t,X^{(2)}_t,\dots,X^{(n)}_t]^\top.
\end{equation}

\subsection{Problem of Causal Network Inference and Challenges}
Given quantitative observations of the dynamic states of individual nodes, often in the form of time series, a central problem is to infer its (causal) system dynamics, which involves the inference of
(1) the causal network topology, $\chi_0(A)$; (2) the link weights, $\{A_{ij}\}$; and (3) the specific forms of functional dependence between nodes, $\{f_i\}$. {These problems are interrelated and all challenging. We focus on the first problem}: inferring the causal network topology $\chi_0(A)$, which serves as the skeleton of the actual network dynamics. See Fig.~\ref{newfig:1} as a schematic illustration.
In particular, the problem of causal network inference can be casted mathematically as:
\begin{equation}
\begin{cases}
\mbox{Given:}&\mbox{Samples of the node states $x^{(i)}_t$ ($i=1,2,\dots,n;~t=1,2,\dots,T$).}\\
\mbox{Goal:}&\mbox{Infer the structure of the underlying causal network,}\\
&\mbox{i.e., find}~\operatorname{argmin}_{\hat{A}}\|\chi_0(A)-\hat{A}\|_0,~
\mbox{where }\|M\|_0\equiv \sum_{i,j}|M_{ij}|^0.
\end{cases}
\end{equation}

One key challenge is that in many applications, the number of nodes $n$ is often {\it large} (usually hundreds at least),  but the sample size $T$ is much smaller than needed for reliable estimation of the $(n\times d)$-dimensional joint distribution. 
We propose that a practical causation inference method should fulfill the following three {requirements}:
\begin{enumerate}
\item {\it Model-free.} The method should not rely on assumptions about either the form or parameters of a model that underlie the process.
\item {\it Computational Efficient.} The method should be computationally efficient.
\item {\it Data Efficient.} The method should achieve high accuracy with relatively small number of samples (i.e., convergence in probability needs to be fast).
\end{enumerate}
In this paper we address the model-free requirement by utilizing information-theoretic measures, and in particular, by using causation entropy. On the other hand, our theoretical developments of the optimal causation entropy principle enables us to develop algorithms that are both computationally efficient and data efficient.

\subsection{Markov Assumptions}
We study the system in a probabilistic framework {assuming stationarity and existence of a continuous distribution}.
We {further} make the following assumptions regarding the conditional distributions $p(\cdot|\cdot)$ arising from the {stationary process given by} Eq.~\eqref{eq:generalnonlinear}. For every node $i\in\mathcal{V}$ and time indices $t,t'$:
\begin{equation}\label{eq:processconds}
\begin{cases}
\mbox{(1) {\em Temporally Markov}:~}\\
\quad\quad\quad\quad p(X_{t}|X_{t-1},X_{t-2},\dots) = p(X_{t}|X_{t-1})= p(X_{t'}|X_{t'-1}).\\
\mbox{(2) {\em Spatially Markov}:~}\\
\quad\quad\quad\quad p(X^{(i)}_{t}|X_{t-1})=p(X^{(i)}_t|X^{(N_i)}_{t-1}).\\
\mbox{(3) {\em Faithfully Markov}:~}\\
\quad\quad\quad\quad{p(X^{(i)}_t|X^{(K)}_{t-1})\neq p(X^{(i)}_t|X^{(L)}_{t-1})\mbox{~whenever~} (K\cap N_i)\neq (L\cap N_i).}
\end{cases}
\end{equation}
Throughout the paper, the relationship between two probability density functions $p_1$ and $p_2$ are denoted as $``p_1=p_2"$ iff they equal {\it almost everywhere}, and $``p_1\neq p_2"$ iff there is a set of positive measure on which the two functions do not equal.

In~Eq.~\eqref{eq:processconds}, condition (1) states that {the underlying dynamics is a time-invariant Markov process\footnote{{If the process is Markov but with higher order, our approach is to convert it into a first-order one as illustrated in Appendix~\ref{app:A} and then apply the theory and algorithms in the main body of the paper to the resulting first-order process.}}.} 
Condition (2) is often referred to as the (local) Markov property~\cite{Lauritzen1996}, which we call Spatially Markov here to differ from Temporally Markov. This condition guarantees that in determining the future state of a node, if knowledge about the past states of all its {causal parents} $N_i$ [as defined in Eq.~\eqref{eq:causalparents}] is given, information about the past of any other node becomes irrelevant. 
{Finally, condition (3) ensures that the set of causal parents is unique and that every causal parent presents an observable effect regardless of the information about  other causal parents\footnote{{Note that without condition (3), the ``true positive" statement in Theorem~\ref{thm:cse} is no longer valid. One simple example is given in Appendix~\ref{app:B} to illustrate this point.}}.}

{The conditional independence between two random variables $X$ and $Y$ given $Z$ is denoted by $(X\ci Y~|~Z )$, i.e.,}
\begin{equation}
	(X\ci Y~|~Z )~\Longleftrightarrow~p(X|Y,Z) = p(X|Z).
\end{equation}
The following results regarding conditional independence will be useful in later sections and are direct consequences of the basic axioms of probability theory~\cite{Grimmett,Lauritzen1996,Pearl2009}:
\begin{equation}\label{eq:condind}
\begin{cases}
\mbox{Symmetry:~}(X\ci Y~|~Z)~\Longleftrightarrow~(Y\ci X~|~Z).\\
\mbox{Decomposition:~} (X\ci YW~|~Z)~\Longrightarrow~(X\ci Y~|~Z).\\
\mbox{Weak union:~} (X\ci YW~|~Z)~\Longrightarrow~(X\ci Y~|~ZW).\\
\mbox{Contraction:~} (X\ci Y~|~Z)\wedge(X\ci W~|~ZY)~\Longrightarrow~(X\ci YW~|~Z).\\
\mbox{Intersection:~} (X\ci Y~|~ZW)\wedge(X\ci W~|~ZY)~\Longrightarrow~(X\ci YW~|~Z).
\end{cases}
\end{equation}
Here $``\wedge"$ denotes the logical operations ``and" (the symbol $``\vee"$ is used later for ``or"),
and $YW$ denotes a joint random variable of $Y$ and $W$.

\subsection{Causation Entropy as an Information-Theoretic Measure}
We review several fundamental concepts in information theory, leading to causation entropy, a model-free information-theoretic statistic that can be used to infer direct causal relationships~\cite{Sun2013PhysicaD}.

Originally proposed by Shannon as a measure of uncertainty and complexity, the (differential) {\em entropy} of a continuous random variable $X\in\mathbb{R}^n$ is defined as~\cite{Cover2006,Shannon1948}\footnote{We follow the convention in Ref.~\cite{Cover2006} to use $h(\cdot)$ for the entropy of a continuous random variable and reserve $H(\cdot)$ for the entropy of a discrete random variable. In the discrete case, we need to replace the integral by summation and probability density by probability mass function in the definition.}
\begin{equation}
	h(X) = -\int p(x)\log{p(x)}dx,
\end{equation}
where $p(x)$ is the probability density function of $X$.
The joint and conditional entropies between two random variables $X$ and $Y$ are defined as [also see Fig.~\ref{newfig:2}(a)]
\begin{equation}
\begin{cases}
	\mbox{Joint entropy:~}h(X,Y)\equiv h(Y,X)\equiv-\int p(x,y)\log p(x,y)dxdy.\\
	\mbox{Conditional entropies:~}
	\begin{cases}
		h(X|Y)\equiv-\int p(x,y)\log p(x|y)dxdy;\\
		h(Y|X)\equiv-\int p(x,y)\log p(y|x)dxdy.
	\end{cases}
\end{cases}
\end{equation}
For more than two random variables, the entropies are similarly defined (as above) by grouping the variables into two classes, one acting as $X$ and the other as $Y$.

The {\it{mutual information}} between two random variables $X$ and $Y$ (conditioning on $Z$) can be interpreted as a measure of the deviation from independence between $X$ and $Y$ (conditioning on $Z$). The corresponding unconditioned and conditional mutual information are defined respectively as
\begin{equation}
\begin{cases}
	\mbox{Mutual information:}~I(X;Y) \equiv h(X)-h(X|Y) \equiv h(Y)-h(Y|X).\\
	\mbox{Conditional mutual information:}~\\
	\quad\quad\quad\quad I(X;Y|Z) \equiv h(X|Z)-h(X|Y,Z) \equiv h(Y|Z)-h(Y|X,Z).
\end{cases}
\end{equation}
The mutual information among three variables $X$, $Y$, and $Z$ is defined as\footnote{{This quantity is often referred to as {\it interaction information}~\cite{McGill1954} or {\it co-information}~\cite{Bell2003}. Another multivariate generalizations of mutual information is total correlation~\cite{Watanabe1960} (also known as multivariate constraint~\cite{Garner1962} or multi-information~\cite{Studeny1998}).}}
\begin{equation}
	I(X;Y;Z) \equiv I(X;Y) - I(X;Y|Z) \equiv I(Y;Z) - I(Y;Z|X) \equiv I(X;Z) - I(X;Z|Y),
\end{equation}
The mutual information between two variables is always nonnegative, $I(X;Y)\geq 0$, with equality if and only if $X$ and $Y$ are independent. Similarly, $I(X;Y|Z)\geq 0$, with equality if and only if $X$ and $Y$ are independent when conditioned on $Z$.
Interestingly, for three or more variables, such an inequality does not hold: the mutual information $I(X;Y;Z)$ can be either positive, negative or zero~\cite{McGill1954}.
Figure~\ref{newfig:2}(a) visualizes the relationships between entropy, conditional entropy, and mutual information.

\begin{figure}[tbp]
\centering
\includegraphics*[width=1\textwidth]{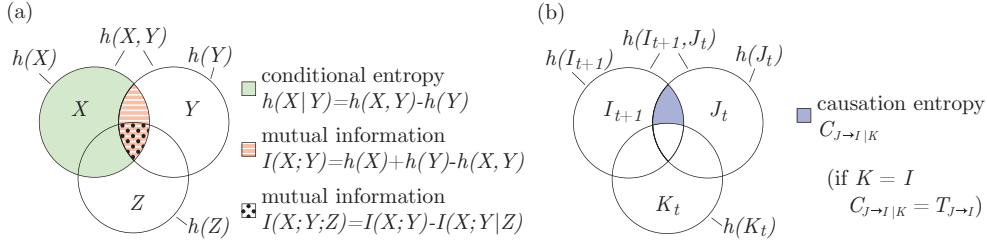}
\caption{
Venn-like diagrams for information-theoretic measures.
(a) Visualization of the relationships between entropy, conditional entropy, and mutual information.
(b) Visualization of the relationships between conditional entropy, causation entropy, and transfer entropy.
In the picture of (b), letters $I$, $J$, and $K$ are used to denote $X^{(I)}$, $X^{(J)}$, and $X^{(K)}$, respectively.
}\label{newfig:2}
\end{figure}

To measure the directionality of information flow between two random processes, Schreiber proposed a specific type of conditional mutual information
called {\it transfer entropy}~\cite{Schreiber2000}. 
{For a stationary first-order Markov process such as the one given by Eq.~\eqref{eq:generalnonlinear}, }
the transfer entropy from $j$ to $i$ can be expressed as
\begin{equation}\label{eq:tedef}
{	T_{j\rightarrow i} \equiv h(X^{(i)}_{t+1}|X^{(i)}_t) - h(X^{(i)}_{t+1}|X^{(i)}_t,X^{(j)}_t),}
\end{equation}
where $h(\cdot|\cdot)$ denotes conditional entropy~\cite{Cover2006}.
Since $h(X^{(i)}_{t+1}|X^{(i)}_t)$ measures the uncertainty of $X^{(i)}_{t+1}$ given information about $X^{(i)}_t$ and
$h(X^{(i)}_{t+1}|X^{(i)}_t,X^{(j)}_t)$ measures the uncertainty of $X^{(i)}_{t+1}$ given information about {\it both $X^{(i)}_t$ and $X^{(j)}_t$},
the transfer entropy $T_{j\rightarrow i}$ can be interpreted as the {\it reduction of uncertainty} about future states of $X^{(i)}$ when the current state of $X^{(j)}$ is provided {\it in addition to that of $X^{(i)}$}. 

Networks of practical interest inevitably contain (many) more than two nodes. As we will show later,  without appropriate conditioning transfer entropy fails to distinguish between direct and indirect causality in networks. To overcome the pairwise limitation of transfer entropy, we define {\em causation entropy}. 
The relationships between entropy, transfer entropy and causation entropy are illustrated in Fig~\ref{newfig:2}(b).

\vspace{0.05in}
\begin{definition}[Causation Entropy~\cite{Sun2013PhysicaD}]
The {\em causation entropy} from the set of nodes $J$ to the set of nodes $I$ conditioning on the set of nodes $K$ is defined as\footnote{ {Note that the definitions in Eq.~\eqref{eq:tedef} and Eq.~\eqref{eq:csedef} can be extended for asymptotically stationary processes by taking the limit of $t\rightarrow\infty$, although the proofs in this paper do not directly apply to such general scenario.}}
\begin{equation}\label{eq:csedef}
{	C_{J\rightarrow I|K} = h(X^{(I)}_{t+1}|X^{(K)}_t) - h(X^{(I)}_{t+1}|X^{(K)}_t,X^{(J)}_t),}
\end{equation}
where $I,J,K$ are all subset of $\mathcal{V}=\{1,2,\dots,n\}$.
In particular, if $J=\{j\}$ and $I=\{i\}$, we simplify the notation as $C_{j\rightarrow i|K}$.
If the conditioning set $K=\varnothing$, we often omit it and simply write $C_{J\rightarrow I}$.
\end{definition}
\vspace{0.05in}

\begin{remark}
Causation entropy is a natural generalization of transfer entropy from measuring pairwise causal relationships to network relationships of many variables. In particular, if $j\in{K}$, then the causation entropy $C_{j\rightarrow i|{K}}=0$ as $j$ does not carry extra information (compared to that of $K$).
On the other hand, if $K=\{i\}$, causation entropy recovers transfer entropy, i.e.,
\begin{equation}
	C_{j\rightarrow i|i}=T_{j\rightarrow i}. 
\end{equation}
Interestingly, in this framework we see that transfer entropy assumes that nodes are self-causal, whereas causation entropy relaxes this assumption.
Preliminary exploration of the differences between the two measures can be found in Ref.~\cite{Sun2013PhysicaD}.
\end{remark}
\vspace{0.05in}

{
\begin{remark}
We note that in addition to Ref.~\cite{Sun2013PhysicaD}, the conditional mutual information between time-lagged variables has been proposed as a statistic for network inference in a few previous studies~\cite{Frenzel2007,Runge2012PRL,Runge2012PRE,Vejmelka2008} (although not referred to as transfer or causation entropy).
\end{remark}
\vspace{0.05in} 
}

\begin{remark}
It seems plausible to conjecture that if two subsets of the nodes satisfy ${K}_1\subset{K}_2$, then $C_{j\rightarrow i|{K}_1}$ would be no less than $C_{j\rightarrow i|{K}_2}$. We remark that this statement about monotonicity is false (see the two examples below).
\vspace{0.05in}

\noindent
{\em Example 1.} Consider the stochastic process
\begin{equation}
	X^{(1)}_t = X^{(2)}_{t-1} + X^{(3)}_{t-1}
\end{equation}
where $X^{(k)}_{t}$ are i.i.d Bernoulli variables: $P(X^{(k)}_t=0)=P(X^{(k)}_t=1)=0.5~(k=2,3)$.
Let $i=1$, $j=2$, ${K}_1=\varnothing$ and ${K}_2=\{3\}$. 
It follows that
\begin{equation}
\begin{cases}
	C_{2\rightarrow 1|\varnothing}=\frac{3}{2}\log{2} - \log{2} = \frac{1}{2}\log{2}\\
	C_{2\rightarrow 1|\{3\}}=\log{2}-0=\log{2}
\end{cases}\\
\Rightarrow~C_{2\rightarrow 1|\varnothing}<C_{2\rightarrow 1|\{3\}}.
\end{equation}

\noindent
{\em Example 2.} Consider the stochastic process 
\begin{equation}
X^{(1)}_{t+1} = X^{(3)}_t,~X^{(2)}_{t+1} = X^{(3)}_t,
\end{equation}
where $X^{(3)}_{t}$ are Bernoulli variables with $P(X^{(3)}_t=0)=P(X^{(3)}_t=1)=0.5$.
Let $i=1$, $j=2$, ${K}_1=\varnothing$ and ${K}_2=\{3\}$.
It follows that
\begin{equation}
\begin{cases}
	C_{2\rightarrow 1|\varnothing}=\log{2}-0=\log{2}\\
	C_{2\rightarrow 1|\{3\}}=0-0=0
\end{cases}
\Rightarrow~C_{2\rightarrow 1|\varnothing}>C_{2\rightarrow 1|\{3\}}.
\end{equation}
\noindent
The seemingly paradoxical observation that $C_{j\rightarrow i|K_1}$ can either be larger or smaller than $C_{j\rightarrow i|K_2}$ despite the fact that ${K}_1\subset{K}_2$ can be understood as follows: When ${K}_1\subset{K}_2$, $C_{j\rightarrow i|{K}_1}-C_{j\rightarrow i|{K}_2}$ corresponds to the mutual information among the three variables $X^{(i)}_{t+1}|X^{({K}_1)}_t$, $X^{(i)}_{t+1}|X^{(j)}_t$ and $X^{(i)}_{t+1}|X^{({K}_2-{K}_1)}_t$ (see Fig.~\ref{newfig:2}). Contrary to the two-variable case where mutual information is always nonnegative, the mutual information among three (or more) variables can either be  positive, negative or zero~\cite{McGill1954}.
\end{remark}

\subsection{Theoretical Properties of Causation Entropy and the Optimal Causation Entropy Principle}
In the following we show that analysis of causation entropy leads to exact network inference for the network stochastic process given by Eq.~\eqref{eq:generalnonlinear} subject to the Markov assumptions in Eq.~\eqref{eq:processconds}. 

We start by exploring basic analytical properties of causation entropy, which is presented as Theorem~\ref{thm:cse} and also summarized in Fig.~\ref{newfig:3}.
\vspace{0.05in}
\begin{theorem}[Basic analytical properties of causation entropy]\label{thm:cse}
Suppose that the network stochastic process given by Eq.~\eqref{eq:generalnonlinear} satisfies the Markov assumptions in Eq.~\eqref{eq:processconds}.
Let $I\subset\mathcal{V}$ be a set of nodes and ${N}_I$ be {its causal parents}.
Consider two sets of nodes $J\subset\mathcal{V}$ and $K\subset\mathcal{V}$. 
The following results hold:
\begin{enumerate}[(a)]
\item (Redundancy) If $J\subset K$, then $C_{J\rightarrow I|K}=0$.
\item (No false positive) If ${N}_I\subset K$, then $C_{J\rightarrow I|K}=0$ for any set of nodes $J$.
\item (True positive) If $J\subset{N}_I$ and $J\not\subset K$, then $C_{J\rightarrow I|K}>0$.
\item (Decomposition) $C_{J\rightarrow I|K}=C_{(K\cup J)\rightarrow I}-C_{K\rightarrow I}$.
\end{enumerate}
\end{theorem}
\begin{proof}
Under the Temporal Markov Condition in Eq.~\eqref{eq:processconds}, there is no time dependence of the distributions.
For notational simplicity we denote the joint distribution $p(X^{(I)}_{t+1}=i,X^{(J)}_t=j,X^{(K)}_t=k)$ by $p(i,j,k)$ and use similar notation for the marginal and conditional distributions.
 It follows that
\begin{eqnarray}\label{eq:hhdiff}
	C_{J\rightarrow I|K} &=& 
	h(X^{(I)}_{t+1}|X^{(K)}_t) - h(X^{(I)}_{t+1}|X^{(K)}_t,X^{(J)}_t)
	= -\int p(i,j,k)\log\Big[\frac{p(i|k)}{p(i|j,k)}\Big]didjdk\nonumber\\
	&\geq& - \log\int p(i,j,k)\frac{p(i|k)}{p(i|j,k)}didjdk~~\mbox{(by {\it Jensen's inequality}~{\cite{Royden}})}\nonumber\\
	&=& -\log\int p(j,k)\frac{p(i,k)}{p(k)}didjdk = -\log(1) = 0,
\end{eqnarray}
where equality holds if and only if $p(i|k)=p(i|j,k)$ {\em almost everywhere}.
The above inequality is also known as the {\em Gibbs' inequality} in statistical physics~\cite{Gibbs}.

To prove $(a)$, we note that $J\subset{K}$ implies that $p(i|k)=p(i|j,k)$ and therefore equality holds (rather than inequality) in Eq.~\eqref{eq:hhdiff}.

To prove $(b)$, it suffices to show that for $J\not\subset K$, $C_{J\rightarrow I|K}=0$.
Since $J\not\subset K$ and ${N}_I\subset K$, based on the Spatial Markov Condition in Eq.~\eqref{eq:processconds}, we have:
\begin{equation}
	p(X^{(I)}_{t+1}|X_t) = p(X^{(I)}_{t+1}|X^{(K\cup J)}_t) = p(X^{(I)}_{t+1}|X^{(K)}_t)=p(X^{(I)}_{t+1}|X^{({N}_I)}_t). 
\end{equation}
Therefore $p(i|j,k)=p(i|k)$ and equality holds in Eq.~\eqref{eq:hhdiff}.

To prove $(c)$, we use the {Faithfully Markov Condition} in Eq.~\eqref{eq:processconds}. Since $J\subset{N}_I$ and $J\not\subset K$, it follows that
\begin{equation}
	p(X^{(I)}_{t+1}|X^{(K)}_t) = p(X^{(I)}_{t+1}|X^{(K\cap N_I)}_t) \neq p(X^{(I)}_{t+1}|X^{(K)}_t,X^{(J)}_t).
\end{equation}
Thus, $p(i|j,k)\neq p(i|k)$ and strictly inequality holds in Eq.~\eqref{eq:hhdiff}. 

Finally, part $(d)$ follows directly from the definition of $C$.\qquad
\end{proof}
\vspace{0.05in}

\begin{figure}[tbp]
\centering
\includegraphics*[width=1\textwidth]{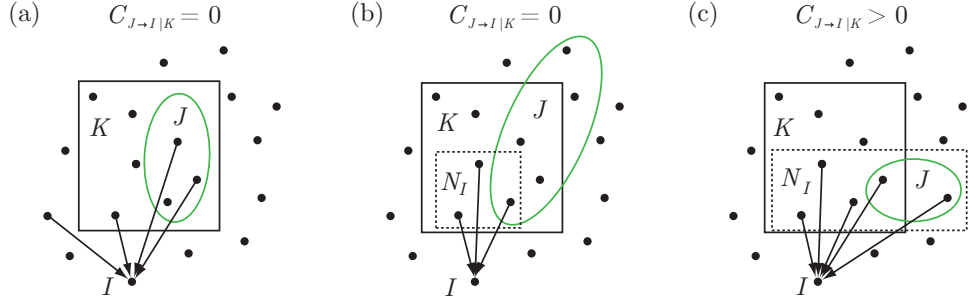}
\caption{
Basic analytical properties of causation entropy (Theorem~\ref{thm:cse}) allowing for the inference of the {causal parents} $N_I$ of a set of nodes $I$.
(a) Redundancy: If $J$ is a subset of the conditioning set $K$ ($J\subset K$), then the causation entropy $C_{J\rightarrow I|K}=0$.
(b) No false positive: If $N_I$ is already included in the conditioning set $K$ ($N_I\subset K$), then $C_{J\rightarrow I|K}=0$.
(c) True positive: If a set $J$ contains at least one {causal parent} of $I$ that does not belong to the conditioning set $K$, i.e.,
$(J\subset{N}_I)\wedge(J\not\subset K)$, then $C_{J\rightarrow I|K}>0$.
}\label{newfig:3}
\end{figure}

Theorem~\ref{thm:cse} allows us to convert the problem of causal network inference into the problem of estimating causation entropy among nodes. In particular, for a given set of nodes $I$, each node $j$ can {\it in principle} be checked independently to determine whether or not it is a {causal parent} of $I$ via either of the following two equivalent criteria (proved in Theorem~\ref{thm:cse2}(a) below)
\begin{equation}\label{eq:condcriteria}
\begin{cases}
\mbox{(1) Node $j\in N_I$ iff there is a set $K\supset N_I$, such that $C_{j\rightarrow I|(K-\{j\})}>0$};\\
\mbox{(2) Node $j\in N_I$ iff for any set $K\subset\mathcal{V}$, $C_{j\rightarrow I|(K-\{j\})}>0$.}
\end{cases}
\end{equation}
Practical application of either criteria to infer large networks is challenging.
Criterion (1) requires a conditioning set $K$ that contains $N_I$ as its subset. Since $N_I$ is generally unknown, one often must use $K=\mathcal{V}$. When the network is large ($n\gg 1$), this requires the estimation of causation entropy for very high dimensional random variables from limited data, which is inherently unreliable~\cite{Runge2012PRL,Runge2012PRE}.
Criterion (2), on the other hand, requires a combinatorial search over all subsets making it computationally infeasible. 

In the following we prove the two inference criteria in Eq.~\eqref{eq:condcriteria}. Furthermore, we show that {\it the {set of causal parents} is the minimal set of nodes that maximizes causation entropy}, which we refer to as the {\it optimal causation entropy principle}.

\vspace{0.05in}
\begin{theorem}[Optimal causation entropy principle for causal network inference]\label{thm:cse2}
Suppose that the network stochastic process given by Eq.~\eqref{eq:generalnonlinear} satisfies the Markov properties in Eq.~\eqref{eq:processconds}.
Let $I\subset\mathcal{V}$ be a given set of nodes and ${N}_I$ be the set of $I$'s causal parents, as defined in Eq.~\eqref{eq:causalparents2}.
It follows that
\begin{enumerate}[(a)]
\item (Direct inference) Node $j\in N_I$ iff $\Leftrightarrow \exists K\supset N_I$ such that $C_{j\rightarrow I|(K-\{j\})}>0\Leftrightarrow \forall K\subset\mathcal{V}, C_{j\rightarrow I|(K-\{j\})}>0$.
\item (Partial conditioning removal) If there exists $ K\subset\mathcal{V}$ such that $C_{j\rightarrow I|(K-\{j\})}=0$, then $j\notin N_I$.
\item (Optimal causation entropy principle) The set of {causal parents} is the minimal set of nodes with maximal causation entropy. \\
Define the family of sets with maximal causation entropy as 
\begin{equation}
	\mathcal{K}=\{K|\forall K'\subset\mathcal{V}, C_{K'\rightarrow I}\leq C_{K\rightarrow I}\}. 
\end{equation}
Then the set of {causal parents} satisfies
\begin{equation}
	N_I = \cap_{K\in\mathcal{K}}K = {\operatorname{argmin}}_{K\in\mathcal{K}}K.
\end{equation}
\end{enumerate}
\end{theorem}
\begin{proof}
First we prove part $(a)$. If $j\in N_I$, then for every $K\subset\mathcal{V}$, $C_{j\rightarrow I|(K-\{j\})}>0$ following Theorem~\ref{thm:cse}(c). This proves both ``$\Rightarrow$".
On the other hand, suppose that $\forall K\subset\mathcal{V}, C_{j\rightarrow I|(K-\{j\})}>0$, then for $K=\mathcal{V}\supset N_I$, it follows that
$C_{j\rightarrow i|(\mathcal{V}-\{j\})}>0$. Node $j\in N_I$ since otherwise $(\mathcal{V}-\{j\})\supset N_I$ which would imply that $C_{j\rightarrow i|(\mathcal{V}-\{j\})}=0$ from Theorem~\ref{thm:cse}(b). Therefore, the two ``$\Leftarrow$"s are also proven.

Next, part $(b)$ follows directly from the contrapositive of Theorem~\ref{thm:cse}(c).

Finally, we prove part $(c)$.
Note that if $N_I\not\subset K$, then $J=N_I-K\neq\varnothing$, and so
$C_{(K\cup J)\rightarrow I}-C_{K\rightarrow I}=C_{J\rightarrow I|K}>0$. Therefore, $K\in\mathcal{K}\Rightarrow N_I\subset K$.
This implies $N_I\subset\cap_{K\in\mathcal{K}}K$.
On the other hand, if $\exists j\in\cap_{K\in\mathcal{K}}K$ with $j\notin N_I$.
Let $K\in\mathcal{K}$ and $L=K-\{j\}$. Since $j\notin N_I$, we have $N_I\subset L\subset K$, and therefore
$C_{K\rightarrow I}-C_{L\rightarrow I}=C_{j\rightarrow I|L}=0$, where the second equality follows from Theorem~\ref{thm:cse}(c).
This shows that $L\in\mathcal{K}$, contradicting with $j\in\cap_{K\in\mathcal{K}}K$.
So $j\in\cap_{K\in\mathcal{K}}K\Rightarrow j\in N_I$, which implies that $\cap_{K\in\mathcal{K}}K\subset N_I$.
Since $\mathcal{K}$ is finite, it follows that  $\cap_{K\in\mathcal{K}}K = {\operatorname{argmin}}_{K\in\mathcal{K}}K$.
\qquad
\end{proof}
\vspace{0.05in}

Based on the optimal causation entropy principle, it seems straightforward to solve the minimax optimization for the inference of $N_I$ by enumerating all subsets of $\mathcal{V}$ with increasing cardinality (starting from $\varnothing$), and terminating when a set $K$ is found to be have maximal causation entropy among all subsets of cardinality $|K|+1$ (i.e., adding any node $j$ to set $K$ does not increase the causation entropy $C_{K\to I }$). Based on Theorem~\ref{thm:cse2}, the set $K=N_I$.
However, this brute-force approach requires $\mathcal{O}(n^{|N_I|})$ causation entropy evaluations, which is computationally inefficient and therefore infeasible for the inference of real world networks which often contain large number of nodes ($n\gg 1$). Such limitation is removed only when the number of {causal parents} is moderately small, $|N_I|=\mathcal{O}(1)$.
In the following section we develop additional theory and algorithms to efficiently solve this minimax optimization problem for causal network inference.

\subsection{Computational Causal Network Inference}
Algorithmically, causal network inference via the optimal causation entropy principle should require as few computations as necessary (computational efficiency) and as few data samples as possible while retaining accuracy (data efficiency).
We introduce two such algorithms that jointly infer the causal network. 
For a given node $i$, the goal is to infer its {causal parents}, as illustrated by nodes in the shaded region of Fig.~\ref{fig23}(a).
Algorithm~\ref{alg:01} aggregatively identifies nodes that form a superset of the {causal parents}, $K\supset N_i$ (proven by Lemma~\ref{thm:cse3}, illustrated in Fig.~\ref{fig23}(b)). Start from a set $K\supset N_i$, Algorithm~\ref{alg:02} prunes away {non-causal nodes} from $K$ leaving only the {causal parents} $N_i$ (proven by Lemma~\ref{thm:cse4}, illustrated in Fig.~\ref{fig23}(c)).

\begin{figure}[htbp]
\centering
\includegraphics*[width=1\textwidth]{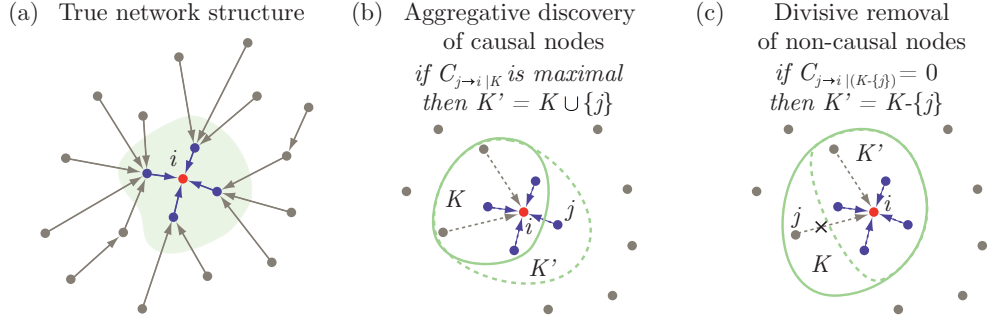}
\caption{
Causal network inference by optimal causation entropy.
(a) {Causal parents and non-causal nodes} of a node $i$. Causal network inference corresponds to identifying the {causal parents} $N_i$ (nodes in shaded region) for every node $i\in\mathcal{V}$.
(b) Nodes are added to the set $K$ in an aggregative fashion, maximizing causation entropy at each step (see Algorithm~\ref{alg:01}).
(c) Starting from a set $K\supset N_i$ ($K$ obtained by Algorithm~\ref{alg:01}), non-causal nodes are progressively removed from $K$ if their causation entropy to node $i$ conditioned on the rest of $K$ is zero (see Algorithm~\ref{alg:02}).
}
\label{fig23}
\end{figure}

\begin{lemma}[Aggregative Discovery of Causal Nodes]\label{thm:cse3}
Suppose that the network stochastic process given by Eq.~\eqref{eq:generalnonlinear} satisfies the Markov properties in Eq.~\eqref{eq:processconds}.
Let $I\subset\mathcal{V}$ and ${N}_I$ be {its causal parents}.
Define the sequences of numbers $\{x_1,x_2,\dots\}$, nodes $\{p_1,p_2,\dots\}$, and nested sets $\{K_0,K_1,K_2,\dots\}$ as: $K_0=\varnothing$, and
\begin{equation}
\begin{cases}
	x_{i} = \max_{x\in(\mathcal{V}-K_{i-1})}C_{x\rightarrow I|K_{i-1}},\\
	p_{i} = {\operatorname{argmax}}_{x\in(\mathcal{V}-K_{i-1})}C_{x\rightarrow I|K_{i-1}},\\
	K_{i} = \{p_1,p_2,\dots,p_{i}\}
\end{cases}
\end{equation}
for every $i\geq1$.
 There exists a number $q$, with $|N_I|\leq q\leq n$, such that
\begin{enumerate}[(a)]
\item The numbers $x_i>0$ for $1\leq i\leq q$ and $x_i=0$ for $i>q$.
\item The set of {causal parents} $N_I\subset K_{q}=\{x_1,x_2,\dots,x_q\}$.
\end{enumerate}
\end{lemma}
\begin{proof}
If $N_I=\varnothing$, the lemma holds trivially. Suppose that $|N_I|\geq 1$ and so $x_1>0$.

To prove $(a)$, we define $q\equiv\min_{x_i=0}(i-1)$ (if all $x_i>0$, define $q\equiv n$). By construction, $x_i>0$ when $i\leq q$ and
$x_{q+1}=0$. This implies that $N_I\subset K_q$ since otherwise there is a node $j$ with $C_{j\rightarrow I|K_q}>0\Rightarrow x_{q+1}>0$. For any $i>q$,
$N_I\subset K_q\subset K_{i-1}$, and thus $C_{j\rightarrow I|K_{i-1}}=0$ for all $j\in(\mathcal{V}-K_{i-1})$, which implies that $x_i=0$.

To prove $(b)$, we note that if there is a node $j\in N_I$ such that $j\notin K_q$, then by the definition of $x_i$ and Theorem~\ref{thm:cse}(c), it follows that
$x_{q+1}\geq C_{j\rightarrow I|K_q}>0$. This is in contradiction with the fact that $x_{i}=0$ for all $i>q$. Therefore, $N_I\subset K_q$.
\qquad
\end{proof}
\vspace{0.05in}

\begin{algorithm}[t]
\caption{{\bf Aggregative Discovery of Causal Nodes}}
\label{alg:01}
\begin{algorithmic}[1]
\REQUIRE Set of nodes $I\subset\mathcal{V}$
\ENSURE $K$ (which will include $N_I$ as its subset)
\STATE Initialize: $K\gets\varnothing$, $x\gets\infty$, $p\gets\varnothing$.
\WHILE{$x>0$}
\STATE $K\gets K\cup\{p\}$
\FOR{every $j\in(\mathcal{V}-K)$} 
\STATE{$x_j\gets C_{j\rightarrow I|K}$}
\ENDFOR
\STATE $x\gets\max_{j\in(\mathcal{V}-K)}x_j$, $p\gets{\operatorname{argmax}_{j\in(\mathcal{V}-K)}}x_j$
\ENDWHILE
\end{algorithmic}
\end{algorithm}

Algorithm~\ref{alg:01} recursively constructs the set $K_q\supset N_I$ (further denoted as $K$) as described by Lemma \ref{thm:cse3} and illustrated in Fig.~\ref{fig23}(b).
To remove indirect and spurious nodes in $K$ that do not belong to $N_I$, we apply the result of Theorem~\ref{thm:cse}(c), $C_{j\rightarrow I|(K-\{j\})}=0\Rightarrow j\notin N_I$. This gives rise to Lemma~\ref{thm:cse4} and Algorithm~\ref{alg:02}.

\vspace{0.05in}
\begin{lemma}[Progressive Removal of Non-Causal Nodes]\label{thm:cse4}
Suppose that the network stochastic process given by Eq.~\eqref{eq:generalnonlinear} satisfies the Markov properties in Eq.~\eqref{eq:processconds}.
Let $I\subset\mathcal{V}$ and ${N}_I$ be {its causal parents}.
Let $K=\{p_1,p_2,\dots,p_q\}$ such that $K\supset N_I$. Define the sequence of sets $\{K_0,K_1,K_2,\dots,K_q\}$ by $K_0=K$, and
\begin{equation}
K_i=
\begin{cases}
	K_{i-1}, & \mbox{if $C_{p_i\rightarrow I|(K_{i-1}-\{p_i\})}>0$};\\
	K_{i-1}-\{p_i\}, & \mbox{if $C_{p_i\rightarrow I|(K_{i-1}-\{p_i\})}=0$}.
\end{cases}
\end{equation}
for every $1\leq i\leq q$. Then $K_q=N_I$.
\end{lemma}
\begin{proof}
By definition, $K_0=K\supset N_I$. We prove that $K_q\supset N_I$ by induction. Suppose that $K_{i-1}\supset N_I$.
If node $p_i\in N_I$, then $C_{p_i\rightarrow I|(K_{i-1}-\{p_i\})}>0$ by Theorem~\ref{thm:cse}(c) and therefore $K_i=K_{i-1}\supset N_I$.
If node $p_i\notin N_I$, then $K_i\supset K_{i-1}-\{p_i\}\supset N_I$.

Next we prove that $K_q\subset N_I$. Suppose that node $p_i\notin N_I$.
Since $K_{i-1}\supset N_I$, the causation entropy $C_{p_i\rightarrow I|(K_{i-1}-\{p_i\})}=0$ by Theorem~\ref{thm:cse}(b), and so
$K_i=K_{i-1}-\{p_i\}$. Therefore, $p\notin K_i\supset K_q$, which implies that $K_q\subset N_I$ (contrapositive).
\qquad
\end{proof}
\vspace{0.05in}

Algorithm~\ref{alg:02} iteratively removes {nodes that are not causal parents} from a set $K$ until the set converges to $N_I$ as described by Lemma~\ref{thm:cse4} and illustrated in Fig.~\ref{fig23}(c). 

\begin{algorithm}[t]
\caption{{\bf Progressive Removal of Non-Causal Nodes}}
\label{alg:02}
\begin{algorithmic}[1]
\REQUIRE Sets of nodes $I\subset\mathcal{V}$ and $K\subset\mathcal{V}$
\ENSURE$ \hat{N}_I$ (inferred set of {causal parents} of $I$)
\FOR{every $j\in K$}
\IF{$C_{j\rightarrow I|{(K-\{j\})}}=0$}
\STATE{$K\gets K-\{j\}$}
\ENDIF
\ENDFOR
\STATE{$\hat{N}_I\gets K$}
\end{algorithmic}
\end{algorithm}

Jointly, Algorithms~\ref{alg:01} and~\ref{alg:02} can be applied to identify the {causal parents} of each node, thus inferring the entire causal network\footnote{Numerically estimated causation entropy is always positive due to finite sample size and numerical precision. 
In practice, one needs to use a statistical test (e.g., permutation test as described in Section 4) to examine
the conditions $x>0$ in Algorithm~\ref{alg:01} and $C_{j\rightarrow I|(K-\{j\})}=0$ in Algorithm~\ref{alg:02}.}.

{
\begin{remark}
There exists a number of algorithms for the problem of network inference, and we will comment on two most relevant techniques. First, we note that the ARACNE algorithm~\cite{Margolin2004} attempts to infer a (non-causal) interaction network based on mutual information. The ARACNE algorithm first computes the mutual information between all pairs of nodes/variables, filtering out the nonsignificant ones, and then enumerates through all triplets and removes links based on the data processing inequality. It was proven to correctly infer the undirected network under the assumptions that (i) mutual information are estimated without error, and (ii) the network is a tree~\cite{Margolin2004}.
Second, the PC algorithm developed by Spirtes, Glymour, and Scheines removes non-causal links by potentially testing all combinations of conditioning subsets, and was proven to correctly infer general causal networks if the conditional independence between the variables can be perfectly examined~\cite{Spirtes2000}. Runge~{\it et.~al.} \cite{Runge2012PRL,Runge2012PRE} recently utilized the PC algorithm to infer causal networks by establishing the conditional dependence/independence via estimation of appropriately defined conditional mutual information between time-lagged variables. We note that whereas we utilize Algorithm~\ref{alg:02} for the divisive step in network inference, an alternative would be to utilize the PC algorithm for the divisive step. Although the accuracy versus efficiency tradeoff for such a modification has yet to be tested, we expect that it may be helpful specifically for inferring the causal parents for nodes with large degree, suggesting that in practical applications one may wish to switch back and forth between Algorithm~\ref{alg:02} and the PC Algorithm for the divisive step, depending on a node's degree.
\end{remark}}


\section{Application to Gaussian Process: Analytical Results}
In this section we make analytical comparison among three approaches to causal network inference: causation entropy, transfer entropy \cite{Schreiber2000}, and conditional Granger causality \cite{Granger1969,Granger1988}. The next section will be devoted to the exploration of the numerical properties of these approaches for general random networks. 

While information-theoretic approaches including causation entropy do not require stringent model assumptions, a linear model must be assumed to offer a fair comparison with the conditional Granger causality. As a benchmark example, we focus on the following linear discrete stochastic network dynamics
\begin{equation}\label{eq:main}
	X^{(i)}_t = \sum_{j\in{N}_i}A_{ij}X^{(j)}_{t-1} + \xi^{(i)}_t~\big(\mbox{or in matrix form:~} X_t=AX_{t-1} + \xi_t\big).
\end{equation}
Here $X^{(i)}_t\in\mathbb{R}$ represents the state of node $i$ at time $t$ ($i\in\{1,2,\dots,n\}, t\in\mathbb{N}$), $\xi^{(i)}_t\in\mathbb{R}$ represents noise, and $A_{ij}X^{(j)}_{t-1}$ models the influence of node $j$ on node $i$. Equation~\eqref{eq:main} finds application in a broad range of areas, including time series analysis (as a multivariate linear autoregressive process~\cite{Brockwell2005}), information theory (as a network communication channel~\cite{Cover2006}), and nonlinear dynamical systems (as a linearized stochastic perturbation around equilibrium states~\cite{LasotaMackey}). 
It is straightforward to check that Eq.~\eqref{eq:main} is a special case of the general network stochastic process, Eq.~\eqref{eq:generalnonlinear}, and asymptotically (as $t\rightarrow\infty$) satisfied the Markov assumptions in Eq.~\eqref{eq:processconds}.

\subsection{Analytical Properties of the Solution}
\subsubsection{Solution Formula}
Defining $X_0=\xi_0$ for convenience, the solution to Eq.~\eqref{eq:main} can be expressed as
\begin{equation}\label{eq:main3}
	X_t = \sum_{k=0}^{t}A^k\xi_{t-k}.
\end{equation}
We assume that $\xi^{(i)}_t$ are i.i.d Gaussian random variables with zero mean and finite nonzero variance, denoted as $\xi^{(i)}_t\sim N(0,\sigma_i^2)$ with $\sigma_i>0$. Therefore,
\begin{equation}\label{eq:noise2}
	\xi_t\sim N(0,S),
\end{equation}
where the covariance matrix $S$ is defined by 
$
	S_{ij}=\delta_{ij}\sigma_i^2
$
with $\delta$ denoting the Kronecker delta.
It follows that
\begin{equation}\label{eq:noise}
\begin{cases}
	\E[\xi^{(i)}_t] = 0, \\
	\cov(\xi^{(i)}_t,\xi^{(j)}_\tau) = \delta_{ij}\delta_{t\tau}.
\end{cases}
\end{equation}

Note that a random variable obtained by an affine transformation of a Gaussian variable is also Gaussian.
For example, if $Y=[Y_1;Y_2]$ is Gaussian, the distribution of $Y_1$ conditioned on $Y_2$ is also Gaussian~\cite{Eaton1983}.
The proposition below follows by expressing random variables via appropriate affine transformations of $\xi_t$'s.

\vspace{0.05in}
\begin{proposition}\label{thm:gaussian}
Let $I$ and $K$ be any subsets of $\mathcal{V}$.
Let $t\in\mathbb{N}$ and $\tau\in\{0\}\cup\mathbb{N}$.
The conditional distribution of $X^{({I})}_{t+\tau}$ given $X^{({K})}_t$ is Gaussian.
\end{proposition}

\subsubsection{Covariance Matrix}
Under an affine transformation from Gaussian variable $Y$ to $Z$ as $Z=CY+d$, the mean and covariance of $Y$ and $Z$ are related by:
$\mu_Z=C\mu_Y+d$ and $\Sigma_Z = C\Sigma_Y C^\top$~\cite{Eaton1983}.
We consider covariance matrices $\Phi(\tau,t)$, where 
the $(i,j)$-th entry of $\Phi(\tau,t)$ is defined as
\begin{equation}
	\Phi(\tau,t)_{ij} \equiv \cov[x^{(i)}_{t+\tau},x^{(j)}_t].
\end{equation}
It follows from Eqs.~\eqref{eq:main3} and~\eqref{eq:noise2} that
\begin{equation}\label{eq:xdist}
	X_t\sim N(0,\Phi(0,t)),~\mbox{where}~\Phi(0,t)=\sum_{k=0}^{t}A^kS(A^k)^\top.
\end{equation}
In the following we prove a sufficient condition for the converge of the covariance matrix $\Phi(0,t)$ as time $t\rightarrow\infty$.
Denote the {\it spectral radius} of a square matrix $M$ by
\begin{equation}
	\rho_M\equiv\max\{|\lambda|:\lambda~\mbox{is an eigenvalue of $M$}\}.
\end{equation}
Note that $\rho_M=\rho_{M^\top}$ since a square matrix and its transpose have the same set of eigenvalues.
For the dynamical system defined by Eq.~\eqref{eq:main}, matrices $A$ with $|\rho_A|<1$ are the only matrices for which the underlying system poses a stable equilibrium in the absence of noise. We refer to these matrices as stable.

\vspace{0.05in}
\begin{definition}[Stable Matrix]
Matrix $M$ is {\em stable} if $\rho_M<1$.
\end{definition}
\vspace{0.05in}

The following is a known result from classical matrix theory~\cite{Horn1985}.

\vspace{0.05in}
\begin{theorem}[Convergence of Matrix Series~\cite{Horn1985}]\label{thm:matrixconv}
The matrix series $\sum_{k=0}^{\infty}M_k$ converges if the scalar series $\sum_{k=0}^{\infty}\|M_k\|$ under any induced norm $\|\cdot\|$ converges.
\end{theorem}
\vspace{0.05in}

Note that it is possible for the matrix series $\sum_{k=0}^{\infty}M_k$ to be convergent while the corresponding scalar series $\sum_{k=0}^{\infty}\|M_k\|$ diverges, analogous to the possibility of a scalar series that is convergent but not absolutely convergent.
Next we state and prove a sufficient condition under which the matrix series in Eq.~\eqref{eq:xdist} converges.

\vspace{0.05in}
\begin{proposition}[Convergence of the Covariance]\label{thm:covconv}
The series $\sum_{k=0}^{\infty}A^kS(A^k)^\top$ converges if $A$ is stable.
\end{proposition}
\begin{proof}
Let $\|\cdot\|$ be any induced norm. Then $\|A^kS(A^k)^\top\| \leq \|A^k\|\cdot\|S\|\cdot\|(A^\top)^k\|$ for any $k\in\mathbb{N}$.
Gelfand's formula (see Ref.~\cite{Gelfand1941}) implies that
\begin{equation}
	\lim_{k\rightarrow\infty}\|A^k\|^{1/k} = \lim_{k\rightarrow\infty} \|(A^\top)^k\|^{1/k} = \rho_A.
\end{equation}
On the other hand, $\lim_{k\rightarrow\infty}\|S\|^{1/k}=1$.
Therefore,
\[ \lim_{k\rightarrow\infty}\|A^kS(A^k)^\top\|^{1/k} \leq \lim_{k\rightarrow\infty}\big(\|A^k\|\cdot\|S\|\cdot\|(A^\top)^k\|\big)^{1/k} = \rho_A^2 < 1,\]
where the last inequality follows from the fact that $A$ is stable. Hence the scalar series $\sum_{k=0}^\infty\|A^k S(A^k)^\top\|_2$ is convergent. 
The proposition follows by Theorem~\ref{thm:matrixconv}.
\qquad
\end{proof}
\vspace{0.05in}

For the remainder of this section, it will be assumed that $A$ is stable in Eq.~\eqref{eq:main}. 
As $t\rightarrow\infty$, 
we drop the second argument in $\Phi(0,t)$ and define the {\it asymptotic covariance matrix}
\begin{equation}\label{eq:cov}
\Phi(0)\equiv\lim_{t\rightarrow\infty}\Phi(0,t)=\sum_{k=0}^{\infty}A^kS(A^k)^\top.
\end{equation}
It follows that $\Phi(0)$ satisfies an algebraic equation given by the proposition below.

\vspace{0.05in}
\begin{proposition}[Asymptotic Covariance Matrix]\label{thm:coveq}
Assume that $A$ is stable. The asymptotic covariance matrix $\Phi(0)=\sum_{k=0}^{\infty}A^kS(A^k)^\top$ satisfies the equation
\begin{equation}\label{eq:dlyap}
	A\Phi(0){A^\top} - \Phi(0) + S = 0.
\end{equation}
\end{proposition}
\vspace{-0.2in}
\begin{proof}
Since $A$ is stable, both of the two matrix series below converge:
\vspace{-0.1in}\[\begin{cases}
	\Phi(0)  = S + ASA^\top + {A^2}S(A^2)^\top + {A^3}S(A^3)^\top + \cdots\\
	A\Phi(0){A^\top} =  ASA^\top + {A^2}S(A^2)^\top + {A^3}S(A^3)^\top + \cdots\\
\end{cases}\vspace{-0.1in}\]
Subtracting the two equations gives the result of the proposition.
\qquad
\end{proof}
\vspace{0.05in}

Equation~\eqref{eq:dlyap} is a (discrete) Lyapunov equation which often appears in stability analysis and optimal control problems~\cite{Rugh1993}. Using ``$\otimes$" as the Kronecker product and ``$\vect$" for the operation of transforming a square matrix to a column vector by stacking the columns of the underlying matrix in order, Eq.~\eqref{eq:dlyap} can be converted into:
\begin{equation}\label{eq:dlyapsol}
	(I_{n^2}-A\otimes{A})\vect(\Phi(0)) = \vect(S),
\end{equation}
where $I_{n^2}$ denotes the identity matrix of size $n^2$-by-$n^2$.
Matrix $\Phi(0)$ can be computed by either solving Eq.~\eqref{eq:dlyap} through iterative methods (see Ref.~\cite{Barraud1977}) or by directly solving Eq.~\eqref{eq:dlyapsol} as a linear system.
In practice, we found the iterative approach to be numerically more efficient and stable compared to direct inversion.

Covariance matrices are in general positive semidefinite~\cite{Eaton1983}. For for the network dynamics defined in Eq.~\eqref{eq:main}, we show that they are indeed positive definite.
\vspace{0.05in}
\begin{proposition}[Positive Definiteness of the Covariance Matrix]
The covariance matrix $\Phi(0,t)$ is positive definite for any $t\in\mathbb{N}$.
The asymptotic covariance matrix $\Phi(0)$ is also positive definite.
\end{proposition}
\begin{proof}
For any unit vector $v\in\mathbb{R}^n$, $v^\top A\Phi(0,0)A^\top v=(A^\top v)^\top A^\top v\geq{0}$.
From Eqs.~\eqref{eq:main3} and~\eqref{eq:noise2}, for any $t\in\mathbb{N}$, $\Phi(0,t)=A\Phi(0,t-1)A^\top + S$. By induction,
\begin{eqnarray}
	v^\top\Phi(0,t)v&=&v^\top{A}\Phi(0,t-1)A^\top{v}+v^\top S{v}\nonumber\\
	&\geq& \left(A^\top{v}\right)^\top\Phi(0,t-1)\left(A^\top{v}\right)+\min_{i}\sigma_i^2\geq{\min_{i}\sigma_i^2}>0.
\end{eqnarray}
This shows that $\Phi(0,t)$ is positive definite (indeed we have: $\rho_{\Phi(0,t)}\geq{\min_{i}\sigma_i^2}>0$).
Taking $t\rightarrow\infty$ in the above estimate also shows that $\Phi(0)$ is positive definite.
\qquad
\end{proof}
\vspace{0.05in}

\subsubsection{Time-Shifted Covariance Matrices}
We define the time-shifted covariance matrix $\Phi(\tau,t)$ for each $t\in\mathbb{N}$ (time) and $\tau\in\mathbb{N}$ (positive time shift between states).
If $A$ is stable, then the covariance matrix $\Phi(t,\tau)$ converges for each time shift $\tau$ as $t\rightarrow\infty$. The (asymptotic) covariance matrices with different time shifts are related by a simple algebraic equation given in the following proposition.

\vspace{0.05in}
\begin{proposition}[Relationship Between Time-Shifted Covariance Matrices]\label{thm:coveq2}
Assume that $A$ is stable. For each $\tau\in\mathbb{N}$, the following limit exists
	\[\lim_{t\rightarrow\infty}\Phi(\tau,t)=\Phi(\tau),\]
where matrix $\Phi(\tau)$ satisfies
\begin{equation}\label{eq:convmaxeq}
 	\Phi(\tau) = A\Phi(\tau-1) = A^2\Phi(\tau-2) = \dots = A^{\tau}\Phi(0).
\end{equation}
\end{proposition}
\begin{proof}
For every $\tau\in\mathbb{N}$ and $t\in\mathbb{N}$, 
it follows that
\begin{equation}\label{eq:entrywiseconv}
\Phi(\tau,t)_{ij} =  \E\Big[\sum_{k=1}^{n}a_{ik}x^{(k)}_{t+\tau-1}+\xi^{(i)}_{t+\tau},x^{(j)}_t\Big]
		= \sum_{k=1}^{n}a_{ik}\Phi(\tau-1,t)_{kj}.
\end{equation}
Therefore, the matrix $\Phi(\tau,t)$ satisfies
\begin{equation}
 	\Phi(\tau,t) = A\Phi(\tau-1,t) = A^2\Phi(\tau-2,t) = \dots = A^{\tau}\Phi(0,t).
\end{equation}
Taking the limit as $t\rightarrow\infty$ in and making use of the fact that $A$ is stable, we reach the conclusion of the proposition.
\qquad
\end{proof}
\subsection{Analytical Expressions of Causation Entropy}
Here we provide analytical expressions for causation entropy of the Gaussian process described in Eq.~\eqref{eq:main}. Because causation entropy can be interpreted as a generalization of both transfer entropy and conditional Granger causality under the appropriate selection of nodes $i$ and $j$ and the conditioning set $K$, these results also provide analytical expressions for transfer entropy and conditional Granger causality.

\subsubsection{Joint entropy expressions}
Let $\Sigma$ be the covariance matrix of a multivariate Gaussian variable $X\in\mathbb{R}^n$ (i.e., $X\sim N({\boldsymbol \mu},\Sigma)$), it follows that~\cite{Ahmed1989}
\begin{equation}\label{eq:gaussianentropy}
	h(X) = \frac{1}{2}\log[\operatorname{det}(\Sigma)] + \frac{1}{2}n\log(2\pi e).
\end{equation}
Note that the right hand side of the above is actually an upper bound for a general random variable (i.e., the equality $``="$ becomes inequality $``\leq"$~\cite{Cover2006}). Therefore, a Gaussian variable maximizes entropy among all variables of equal covariance.

The random variable $X_t$ is Gaussian and converges to $N(0,\Phi(0))$ as $t\rightarrow\infty$.
For an arbitrary subset of the nodes $K=\{k_1,k_2,\dots,k_\ell\}$. The joint entropy is
\begin{equation}
	h(X^{(K)}) = \lim_{t\rightarrow\infty}h(X^{(K)}_t) = \frac{1}{2}\log(|\Phi_{KK}(0)|) + \log(2\pi e).
\end{equation}
Here we have introduced the notation 
\begin{equation}
	\Phi_{IJ}(0)\equiv P(I)\Phi(0)P(J)^\top,
\end{equation} 
where for a set $K=\{k_1,k_2,\dots,k_\ell\}$, $P(K)$ is the $\ell$-by-$n$ projection matrix defined as
\begin{equation}
	P(K)_{ij}=\delta_{k_i,i}
\end{equation}

\subsubsection{Causation Entropy}For the Gaussian process given by Eq.~\eqref{eq:main}, we obtain the analytical expression of causation entropy as 
\begin{equation}\label{eq:cseformula}
	C_{J\rightarrow I|K} = \frac{1}{2}\log
	\left(\frac{\operatorname{det}\left[\Phi(0)_{II}-\Phi(1)_{IK}\Phi(0)_{KK}^{-1}\Phi(1)_{IK}^\top\right]}
	{\operatorname{det}\left[\Phi(0)_{II}-\Phi(1)_{I,K\cup J}\Phi(0)_{K\cup J,K\cup J}^{-1}\Phi(1)_{I,K\cup J}^\top\right]}\right)
\end{equation}
If $J=\{j\}$ and $I=\{i\}$, this equation simplifies to
\begin{equation}\label{eq:cseformula2}
	C_{j\rightarrow i|K} = \frac{1}{2}\log
	\left(\frac{\Phi(0)_{ii}-\Phi(1)_{iK}\Phi(0)_{KK}^{-1}\Phi(1)_{iK}^\top}
	{\Phi(0)_{ii}-\Phi(1)_{i,K\cup\{j\}}\Phi(0)_{K\cup\{j\},K\cup\{j\}}^{-1}\Phi(1)_{i,K\cup\{j\}}^\top}\right).
\end{equation}

\subsubsection{Transfer Entropy}
Recall that causation entropy recovers transfer entropy when $K=\{i\}$. Letting $K=\{i\}$ in the formula above gives the transfer entropy (with single time lag) for multivariate Gaussian variables:
\begin{eqnarray}\label{eq:te1}
	&~&T_{j\rightarrow i} = C_{j\rightarrow i|i} = \frac{1}{2}\log\Big(1+\frac{\alpha_{ij}}{\beta_{ij} - \alpha_{ij}}\Big),\nonumber\\
	&~&\mbox{where}
	\begin{cases}
	\alpha_{ij}\equiv \big(\Phi(0)_{ii}\Phi(1)_{ij}-\Phi(0)_{ij}\Phi(1)_{ii}\big)^2,\\
	\beta_{ij}\equiv\big(\Phi(0)_{ii}^2-\Phi(1)_{ii}^2\big)\big(\Phi(0)_{ii}\Phi(0)_{jj}-\Phi(0)_{ij}^2\big).
\end{cases}
\end{eqnarray}
It follows that $\beta_{ij}\geq\alpha_{ij}\geq{0}$, and therefore
$T_{j\rightarrow i}\geq 0$ ($T_{i\rightarrow i}=0$). Furthermore, 
\begin{equation}\label{eq:te2}
	T_{j\rightarrow i}=0~\Longleftrightarrow~\alpha_{ij}=0~\Longleftrightarrow~\sum_{k=1}^{n}A_{ik}\big(\Phi(0)_{ii}\Phi(0)_{kj}-\Phi(0)_{ij}\Phi(0)_{ki}\big)=0.
\end{equation}

\subsubsection{Conditional Granger Causality}
As shown in Ref.~\cite{Barnett2009}, when the random variables are Gaussian, expression of Granger Causality is equivalent as that of transfer entropy (and also causation entropy introduced here). In fact, for Gaussian variables, the Granger Causality from $j$ to $i$ without conditioning equals $2C_{j\rightarrow i}$, while the conditional Granger causality (with full conditioning) equals $2C_{j\rightarrow i|(\mathcal{V}-\{j\})}$.
\subsection{Analytical Results for Directed Linear Chain, Directed Loop, and Directed Trees}
We derive expressions of transfer entropy and causation entropy for several classes of networks including directed linear chains, directed loops, and directed trees. These results highlight that although transfer entropy may indicate the direction of information flow between two nodes, its application to causal network inference is often unjustified as it cannot distinguish between direct and indirect causal relationships (unless appropriate conditioning is adopted as in causation entropy).

\subsubsection{Directed Linear Chain}
Denote a directed linear chain of $n$ nodes as
\begin{equation}\label{eq:linearchain}
	1\rightarrow 2\rightarrow 3\dots \rightarrow n.
\end{equation}
For simplicity we assume that all links have the same weight $w=1$. Consequently, the corresponding adjacency matrix $A=[A_{ij}]_{n\times n}$ is given by
\begin{equation}\label{eq:linearchain2}
	A_{ij} = \delta_{i,j+1}.
\end{equation}
It follows that $\rho_A=0$ and therefore $A$ is stable. 
By inverting the lower-triangular matrix $(I_{n^2}-A\otimes{A})$ in Eq.~\eqref{eq:dlyapsol} and applying Eq.~\eqref{eq:convmaxeq}, we obtain that\begin{equation}
\begin{cases}
	\Phi(0)_{ij} = \delta_{ij}\sum_{k=1}^{j}\sigma_k^2,\\
	\Phi(1)_{ij} = \delta_{i,j+1}\sum_{k=1}^{j}\sigma_k^2.
\end{cases}
\end{equation}
Letting $K=\varnothing$ and $K=\{i\}$ respectively in Eqs.~\eqref{eq:cseformula2} and~\eqref{eq:te1}, it follows that
\begin{equation}
C_{j\rightarrow i} = T_{j\rightarrow i} = \frac{1}{2}\delta_{i,j+1}\log\left(1+\frac{\sum_{k=1}^{j}\sigma_k^2}{\sigma_i^2}\right).\label{eq:linear_chain1}
\end{equation}

Therefore, for the directed linear chain defined in Eq.~\eqref{eq:linearchain2}, transfer entropy $T_{j\rightarrow i}=C_{j\rightarrow i}$, and it is positive
if and only if there is a direct link $j\rightarrow i$, i.e.,
\begin{equation}
	C_{j\rightarrow i} = T_{j\rightarrow i}>0~\Leftrightarrow~A_{ij}=1,~~\mbox{and}~~
	C_{j\rightarrow i} = T_{j\rightarrow i}=0~\Leftrightarrow~A_{ij}=0.
\end{equation}
Interestingly, both causation entropy $C_{j\rightarrow{j+1}}$ and transfer entropy $T_{j\rightarrow{j+1}}$ increase monotonically as a function of $j$, and the values only depend on part of the chain from the top node (node $1$) to node $j+1$ and not on the rest of the network.
Interpreting the monotonicity in term of the network structure, the closer node $j$ is to the end of the chain, effectively the more information is transferred through the directed link $j\rightarrow{j+1}$.
Figure.~\ref{newfig:6}(a) illustrates this via a network of $n=1000$ nodes.

\subsubsection{Directed Loop}
Consider now a directed loop with $n$ nodes, denoted as
\begin{equation}\label{eq:directloop}
	1\rightarrow 2\rightarrow 3\dots \rightarrow n\rightarrow 1.
\end{equation}
Let $w>0$ be the uniform link weight. It follows that $\rho_A=w$. Thus, for the adjacency matrix $A$ to be stable, we must have $w<1$.
To keep the symmetry of the problem, we further assume that the variance of noise is the same at each node, therefore
\begin{equation}
	\sigma^2\equiv \sigma_1^2=\sigma_2^2=\dots\sigma_n^2.
\end{equation}
The entries in $\Phi(0,t)$ satisfy
\begin{equation}\label{eq:phi0t}
	\Phi(0,t)_{ij} = w^2\Phi(0,t-1)_{p_i,p_j} + \delta_{ij}\sigma^2,
\end{equation}
where $p_i$ denotes the unique node that directly links to node $i$.
Taking the limit as $t\rightarrow\infty$ and solve the resulting recursive equations, we obtain that for
\begin{equation}\label{eq:loopphi0}
	\begin{cases}
		\Phi(0)_{ij} = \delta_{ij}\sigma^2/(1-w^{2}),\\
		\Phi(1)_{ij} = \delta_{p_i,j}\sigma^2 w/(1-w^{2}).
	\end{cases}
\end{equation}
where the second equation is obtained through $\Phi(0)_{ij}$ and Eq.~\eqref{eq:convmaxeq}.
Letting $K=\varnothing$ and $K=\{i\}$ respectively in Eqs.~\eqref{eq:cseformula2} and~\eqref{eq:te1}, we conclude that
\begin{equation}
C_{j\rightarrow i} = T_{j\rightarrow i} = \frac{1}{2}\delta_{p_i,j}\log\Big(\frac{1}{1-w^2}\Big).\label{eq:directed_loop1}
\end{equation}
Note that causation entropy and transfer entropy equal and do not depend on the noise variation $\sigma^2$, and they are positive if and only if there is a direct link $j\rightarrow i$, i.e.,
\begin{equation}
	C_{j\rightarrow i} = T_{j\rightarrow i}>0~\Leftrightarrow~A_{ij}=1,~~\mbox{and}~~
	C_{j\rightarrow i} = T_{j\rightarrow i}=0~\Leftrightarrow~A_{ij}=0.
\end{equation}
By symmetry, causation entropy and transfer entropy through each directed link is the same. As the link weight $w$ increases in $(0,1)$, both increase monotonically in $(0,\infty)$. The larger the link weight $w$ is, the larger amount of information is transferred via each directed link, as intuitively expected. Also see Fig.~\ref{newfig:6}(b) as an illustration.

\begin{figure}[tbp]
\centering
\includegraphics*[width=1\textwidth]{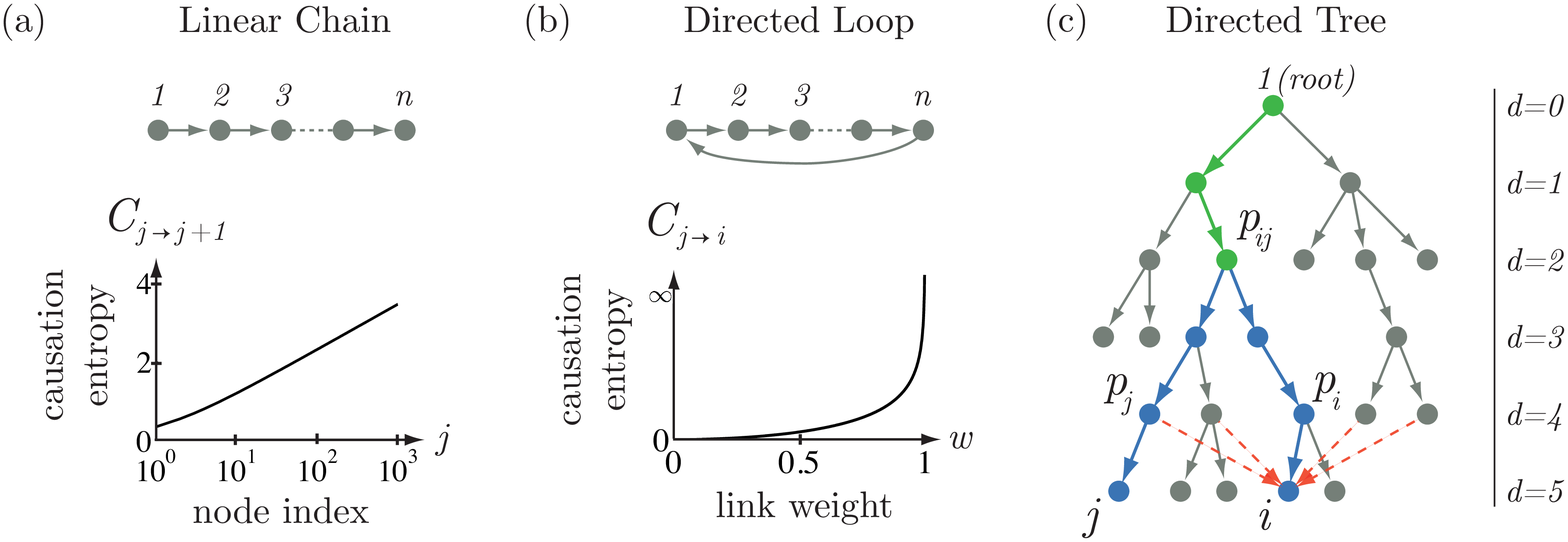}
\caption{
Causation Entropy and transfer entropy for a Gaussian process on three classes of networks. 
(a) For directed linear chains, both causation entropy and transfer entropy correctly identify the network as $C_{j\to i}=T_{j\to i}>0$ iff $i=j+1$ (otherwise $C_{j\to i}=T_{j\to i}=0 $). The dependence of $C_{j\to j+1}$ on node index $j$ is given by Eq.~\eqref{eq:linear_chain1} and plotted.
(b) For directed loops, causation entropy and transfer entropy again correctly identify the network topology with $C_{j\to i}=T_{j\to i}>0$ iff $j\rightarrow i$. The dependence of $C_{j\to i}$ on link weight $w$ is given by Eq.~\eqref{eq:directed_loop1} as shown. 
(c) For directed trees, causation entropy given by Eq.~\eqref{eq:directed_tree1} correctly identifies the network topology based on Eq.~\eqref{eq:directed_tree3}. In contrast, transfer entropy without appropriate conditioning infers many links that do not exist in the actual network (red dashed lines), as described by Eq.~\eqref{eq:directed_tree2}.
}\label{newfig:6}
\end{figure}

\subsubsection{Directed Trees}
We now consider directed tree networks with uniform link weight $w=1$ and unit node variance\footnote{Similar results hold for trees with general link weights and node variances but the corresponding equations are too cumbersome to list.}
\begin{equation}
	\sigma_1^2=\sigma_2^2=\dots\sigma_n^2=1.
\end{equation}
A directed tree has one root (indexed as node $1$ without loss of generality) and each non-root node $i$ ($i\neq{1}$) has exactly one {\it ancestor}, denoted by $p_i$.
The corresponding adjacency matrix $A=[A_{ij}]_{n\times n}$ thus satisfies
\begin{equation}
	A_{ij} = (1-\delta_{i1})\delta_{i,p_i}.
\end{equation}
It can be shown that $\rho_A=0$.
For $i\neq{1}$, we denote the directed path from $1$ to $i$ by
\begin{equation}
	1=p^{(d_i)}_i\rightarrow{p^{(d_i-1)}_i}\rightarrow{}\dots\rightarrow{p^{(1)}_i\equiv p_i}\rightarrow{p^{(0)}_i\equiv i},
\end{equation}
where $d_i$ is the {\it depth} of node $i$ in the tree (for node $1$, we define its depth $d_1=0$).
Thus, the highest node in the tree is the root, and the lowest nodes have the greatest depth.
For any two nodes $(i,j)$, we denote their {\it lowest common ancestor} by $p_{ij}$, i.e.,
\begin{equation}\label{eq:treelca}
p_{ij} = \operatorname*{arg\,max}_{\{k|\exists{\ell,m\geq{0}},s.t.,p_i^{(\ell)}=p_j^{(m)}\}}d_k.
\end{equation}

The covariance matrix $\Phi(0,t)$ satisfies
\begin{equation}
	\Phi(0)_{ij} = \delta_{1i}\delta_{1j}\sigma_1^2 + (1-\delta_{1i})(1-\delta_{1j})[\Phi(0)_{p_i,p_j} + \delta_{ij}].
\end{equation}
We solve these recursive equations to obtain
\begin{equation}\label{eq:dtreephi}
\begin{cases}
\Phi(0)_{ij} = \delta_{d_i,d_j}(d_{p_{ij}}+1)\\
\Phi(1)_{ij} = (1-\delta_{i1})\delta_{d_i,d_j+1}(d_{p_{ij}}+1),
\end{cases}
\end{equation}
where $p_{ij}$ is defined in Eq.~\eqref{eq:treelca} and $\Phi(1)_{ij}$ is obtained by $\Phi(1)=A\Phi(0)$.

We calculate causation entropy and transfer entropy through Eqs.~\eqref{eq:cseformula2} and~\eqref{eq:te1}:
\begin{equation}\label{eq:directed_tree1}
	C_{j\rightarrow i} = T_{j\rightarrow i} = \frac{1}{2}\delta_{d_i,d_{j}+1}\log\frac{(d_i+1)(d_j+1)}{(d_i+1)(d_j+1)-(d_{p_{ij}}+1)^2}.
\end{equation}
Note that  in general $0\leq d_{p_{ij}}\leq\min\{d_i,d_j\}$. Thus $C_{j\rightarrow i}=T_{j\rightarrow i}\leq\frac{1}{2}\log(1+d_i)$, with equality if and only if $j$ is the ancestor of $i$ (i.e., $j=p_i=p_{ij}$).
Therefore, we have
\begin{equation}\label{eq:directed_tree2}
\begin{cases}
	T_{j\rightarrow i}>0~\Leftrightarrow~d_i=d_j+1~\Leftarrow~A_{ij}=1~(\mbox{but}~T_{j\rightarrow i}>0\not\Rightarrow A_{ij}=1);\\
	T_{j\rightarrow i}=0~\Leftrightarrow~d_i\neq d_j+1~\Rightarrow~A_{ij}=0~(\mbox{but}~A_{ij}=0\not\Rightarrow T_{j\rightarrow i}=0).
\end{cases}
\end{equation}
In other words, transfer entropy being positive (without appropriate conditioning) corresponds to a superset of the links that actual exist in a directed tree, and the inferred network using this criterion will potentially contain many false positives. See Fig.~\ref{newfig:6}(c) as an example.
On the other hand, for a given node $i\neq 1$, we have
\begin{equation}\label{eq:directed_tree3}
\begin{cases}
	p_i = \operatorname*{arg\,max}_{j}C_{j\rightarrow i},\\
	C_{j\rightarrow i|\{p_i\}} = 0.
\end{cases}
\end{equation}
Therefore, for each node $i$, the node $j$ that maximizes causation entropy $C_{j\rightarrow i}$ among all nodes is inferred as the {causal parent} of $i$. Conditioned on this node, the causation entropy from any other node to $i$ will become zero, indicating no other directed links to node $i$. This causation entropy based procedure allows for exact and correct inference of the underlying causal network, a directed tree.


\section{Application to Gaussian Process: Numerical Results}
In this section, we illustrate that causal network inference by optimal causation entropy is reliable and efficient for the Gaussian process, Eq.~\eqref{eq:main}, on large random networks. 
\subsection{Random Network Model and Time Series Generation}
We consider signed Erd\H{o}s-R\'enyi networks, which is a generation of its original model~\cite{BollobasBook}. In particular, each network consists of $n$ nodes ($\mathcal{V}=\{1,2,\dots,n\}$), such that each directed link $j\rightarrow i$ is formed independently with equal probability $p$, giving rise to a directed network with approximately $n^2p$ directed links. For generality, we allow the link weight of each link $j\rightarrow i$ to be either positive ({$A_{ij}=w$}) or negative ({$A_{ij}=-w$}), with equal probability. Recalling that the network adjacency matrix $A$ is defined entry-wise by {$A_{ij}\in\{w,-w\}$} iff there exists a directed link $j\rightarrow i$ (otherwise $A_{ij}=0$), the link weight $w$ may be selected to tune the spectral radius $\rho(A)$ of matrix $A$.

We generate time series from the stochastic equation, Eq.~\eqref{eq:main}, where matrix $A$ is obtained from the network model and random variables $\xi_t\sim\mathcal{N}(0,S)$, where the covariance matrix $S$ is taken to be the identity matrix of size $n\times n$. To reduce transient effects, for a given sample size $T$ we solve Eq.~\eqref{eq:main} for $10T$ time steps and only use the final $10\%$ of the resulting time series.

To summarize, our numerical experiments contain parameters: $n$~(network size), $p$~(connection probability), $\rho(A)$~(spectral radius of $A$), and $T$~(sample size).

\subsection{Practical Considerations for Network Inference}\label{sec:practical}
We have established by Theorems \ref{thm:cse} and \ref{thm:cse2} and Lemmas \ref{thm:cse3} and \ref{thm:cse4} that in theory, exact network inference can be achieved by optimal causation entropy, which involves implementing Algorithms \ref{alg:01}~(Aggregative Discovery) and \ref{alg:02}~(Progressive Removal) to correctly identify the set of {causal parents} $N_i$ for each node $i\in\mathcal{V}$. 

In practice, the success of our optimal causation entropy approach (and in fact, any entropy-based approaches) depends crucially on reliable estimation of the relevant entropies in question from data. This leads to two practical challenges.

(1) Entropies must be {\it estimated} from {\it finite} time series data. While there are several techniques for estimating entropies for general multivariate data, the accuracy of such estimations are increasingly inaccurate for small sample sizes and high-dimensional random variables~\cite{Paninski2003}. In this research, we side-step this computational complexity by using knowledge of the asymptotic functional form for the entropy of the Gaussian Process, where the covariance matrices $\Phi(0)$ and $\Phi(1)$ in Eqs.~\eqref{eq:cseformula} and~\eqref{eq:cseformula2} are estimated directly from the time series data.

(2) Application of the theoretical results rely on determining whether the causation entropy $C_{j\rightarrow i|K}>0$ or $C_{j\rightarrow i|K}=0$. However, the estimated value of $C_{j\rightarrow i|K}$ based on sample covariances is necessarily positive given finite sample size and finite numerical precision. Therefore, a statistical test must be used to assess the significance of the observed positive causation entropy.
We here adopt a widely used approach in non-parametric statistics, called the {\it permutation test}\footnote{The idea of a permutation test is to perform (large number of) random permutations of a subset of the data leaving the rest unchanged, giving rise to an empirical distribution of the static of interest. The observed statistic from the original data is then located on this empirical distribution in order to associate its statistical significance~\cite{Good2005}.}. Specifically, we propose the following permutation test based on the null hypothesis that causation entropy $C_{j\rightarrow i|K}=0$: first perform $r$ random (temporal) permutations of the time series $\{X^{(j)}_t\}$, leaving the rest of the data unchanged; we then construct an empirical cumulative distribution $\hat{F}(x)$ of the estimated causation entropy from the permuted time series\footnote{The accuracy of this empirical distribution and therefore the permeation test increases with increasing number of permutations $r$. However, as $r$ increases, the computational complexity also increases, scaling roughly as a linear function of $r$.}; finally, given a prescribed significance level $\theta$, the observed $C_{j\rightarrow i|K}=c$ is declared {\it significant} (i.e., the null hypothesis is rejected at level $\theta$) if $\hat{F}(c)>\theta$.

To summarize, the inference algorithms contain two parameters to be used in the permutation test: $r$~(number of random permutations) and $\theta$~(significance threshold).

\subsection{Comparing Optimal Causation Entropy, Conditional Granger, and transfer entropy}\label{sec:compare}
Here we compare the performance of three approaches of causal network inference: conditional Granger (see for example Ref.~\cite{Gao2011,Guo2008}), transfer entropy (see Ref.~\cite{Vicente2011JCN} and the references therein), and optimal causation entropy (oCSE). In particular, the conditional Granger and transfer entropy approaches under consideration both estimate the entropy $C_{j\rightarrow i|K}$ for each pair of nodes $(i,j)$ independently, with the choice of $K=\mathcal{V}-\{j\}$ in the case of conditional Granger and $K=\{i\}$ in the case of transfer entropy. In both approaches, a causal link $j\rightarrow i$ is inferred if the observed $C_{j\rightarrow i|K}>0$ is assessed as significant under the permutation test. The oCSE approach combines Algorithms~\ref{alg:01} and~\ref{alg:02} and the permutation test is used once per each iteration (line 2 of both algorithms).

The performance of the three approaches are quantified by two types of inference error: false negative ratio, denoted as $\varepsilon_{-}$ and defined as the fraction of links in the original network that are not inferred; and false positive ratio, denoted as $\varepsilon_{+}$ and defined as the fraction of non-existing links in the original networks that are inferred. In terms of the adjacency matrix $A$ of the original network and that of the inferred network $\hat{A}$, these ratios can be computed as
\begin{equation}\label{eq:inferror}
\begin{cases}
\varepsilon_{-}\equiv \dfrac{\mbox{number of $(i,j)$ pairs with $\chi_0(A)_{ij}=1$ and $\chi_0(\hat{A})_{ij}=0$ }}{\mbox{number of $(i,j)$ pairs with $\chi_0(A)_{ij}=1$}},\vspace{0.05in}\\
\varepsilon_{+}\equiv \dfrac{\mbox{number of $(i,j)$ pairs with $\chi_0(A)_{ij}=0$ and $\chi_0(\hat{A})_{ij}=1$ }}{\mbox{number of $(i,j)$ pairs with $\chi_0(A)_{ij}=0$}}.
\end{cases}
\end{equation}

For the random networks considered here, we found that the Algorithm~\ref{alg:01} achieves almost the same accuracy as the combination of Algorithms~\ref{alg:01} and~\ref{alg:02}. We therefore present results which are based on the numerical application of Algorithm~\ref{alg:01} alone, leaving detailed numerical study of Algorithm~\ref{alg:02} to future work.

Figure~\ref{newfig:7}(a-b) shows that although the conditional Granger approach is theoretically correct and works well for small network size with sufficient samples, it suffers from increasing inference error as the network size increases and become extremely inaccurate when the network size $n$ starts to surpass the sample size $T$. Such limitation is overcome by the oCSE approach, where both the false positive and false negative ratios remain close to zero as the network size increases. The reason that oCSE is accurate even as $n$ increases is that it builds the {causal parent} set in an {\it aggregative} manner, therefore relying only on estimating entropy in relatively low dimensions (roughly the same dimension as the number of {causal parents} per node). In sharp contrast, the conditional Granger approach requires the estimation of entropy in the full $n$-dimensional space and therefore requires many (potentially exponentially) more samples to achieve the same accuracy when $n$ becomes large.

\begin{figure}[htbp]
\centering
\includegraphics*[width=1\textwidth]{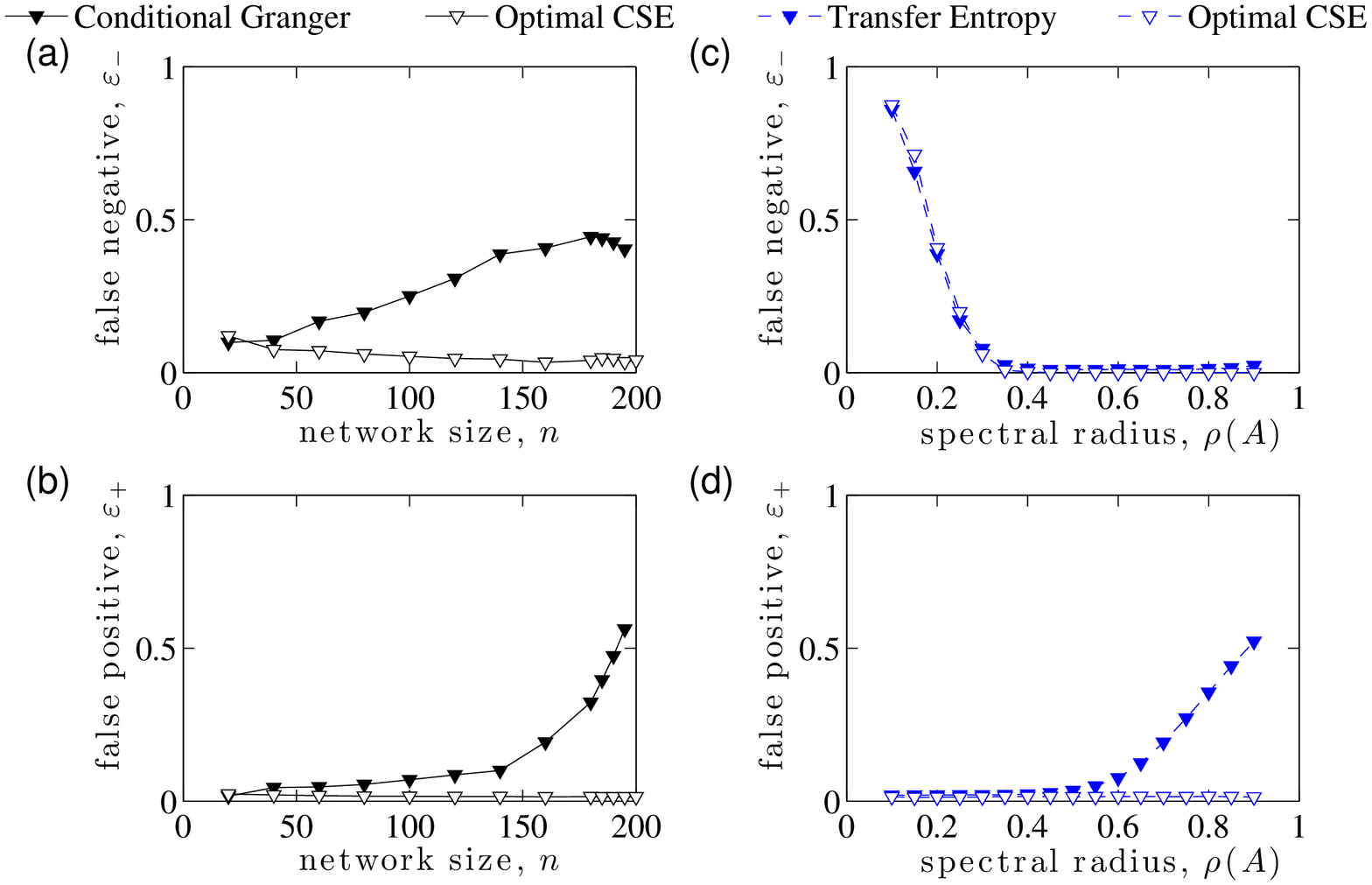}
\caption{
Comparison of causal network inference approaches: conditional Granger, transfer entropy, and oCSE.
The time series are generated from the Gaussian process defined in Eq.~\eqref{eq:main} using signed Erd\H{o}s-R\'enyi networks (see Sec 4.2 for details). Two types of inference error are examined: false negative and false positive ratios, defined in Eq.~\eqref{eq:inferror}.
(a-b) Inference error as a function of network size $n$ using conditional Granger versus oCSE approaches. Here the networks have fixed average degree $np=10$ and spectral radius $\rho(A)=0.8$. Sample size is $T=200$.
(c-d) Inference error as a function of the spectral radius $\rho(A)$ using transfer entropy versus oCSE approaches. Here the networks have fixed number of nodes $n=200$ and average degree $np=10$. Sample size is $T=2000$.
For all three approaches we apply the permutation test using $r=100$ permutations and significance level $\theta=99\%$.
Each data point is obtained from averaging over $20$ independent simulations of the network dynamics, Eq.~\eqref{eq:main}.
}\label{newfig:7}
\end{figure}

Figure~\ref{newfig:7}(c-d) shows that even for a sufficient number of samples, the transfer entropy approach without appropriate conditioning can lead to considerable inference error, and is therefore inherently unsound for causal network inference. In particular, although inference by both transfer entropy and oCSE give similar false negatives in the regime of $\rho(A)\approx 0$ where the dynamics is dominated by noise and not the causal dependences, transfer entropy yields increasing false positives when the causal links dominate, $\rho(A)\rightarrow 1$. This is mainly due to the fact that as $\rho(A)\rightarrow 1$, indirect causal nodes become increasingly difficult to distinguish from direct ones without appropriate conditioning~\cite{Sun2013PhysicaD}. oCSE, on the other hand, consistently yields nearly zero false positive ratios in the entire range of $\rho(A)$. 
Interestingly, the spectral radius $\rho(A)$ can be interpreted as the {\it information diffusion rate} on networks and found to be very close to criticality (i.e., $\rho(A)\approx 1$) in neuronal networks~\cite{Kanuchi2011,Larrremore2011}.

These numerical experiments highlight that whereas the conditional Granger approach is inaccurate for $T\lesssim n$ and the transfer entropy approach is inaccurate when $\rho(A)\lesssim1$, the proposed oCSE approach overcomes both limitations and yields almost exact network inference even for limited sample size.

\subsection{Performance of Optimal Causation Entropy Approach for Causal Network Inference}\label{sec:random}
Having established the advantages of the oCSE approach, we now examine its performance under various parameter settings.

First, we examine the effect of the significance level $\theta$ on the inference error. As shown in Fig.~\ref{newfig:8}(a-b), the false negative ratio $\epsilon_{-}$ does not seem to depend on $\theta$ and converges to zero as sample size $T$ increases. On the other hand, as $T\rightarrow\infty$, the false positive ratio saturates at the level $\epsilon_{+}\sim(1-\theta)$, which is consistent with the implementation of the permutation test which rejects the null hypothesis at $\theta$. This observation suggests that in order to achieve higher accuracy given sufficient sample size, one should choose $\theta$ as close to one as possible. The tradeoff in practice is that reliable implementation using larger $\theta$ requires an increasing number of  permutations and therefore increases the computational complexity of the inference algorithms.

\begin{figure}[htbp]
\centering
\includegraphics*[width=1\textwidth]{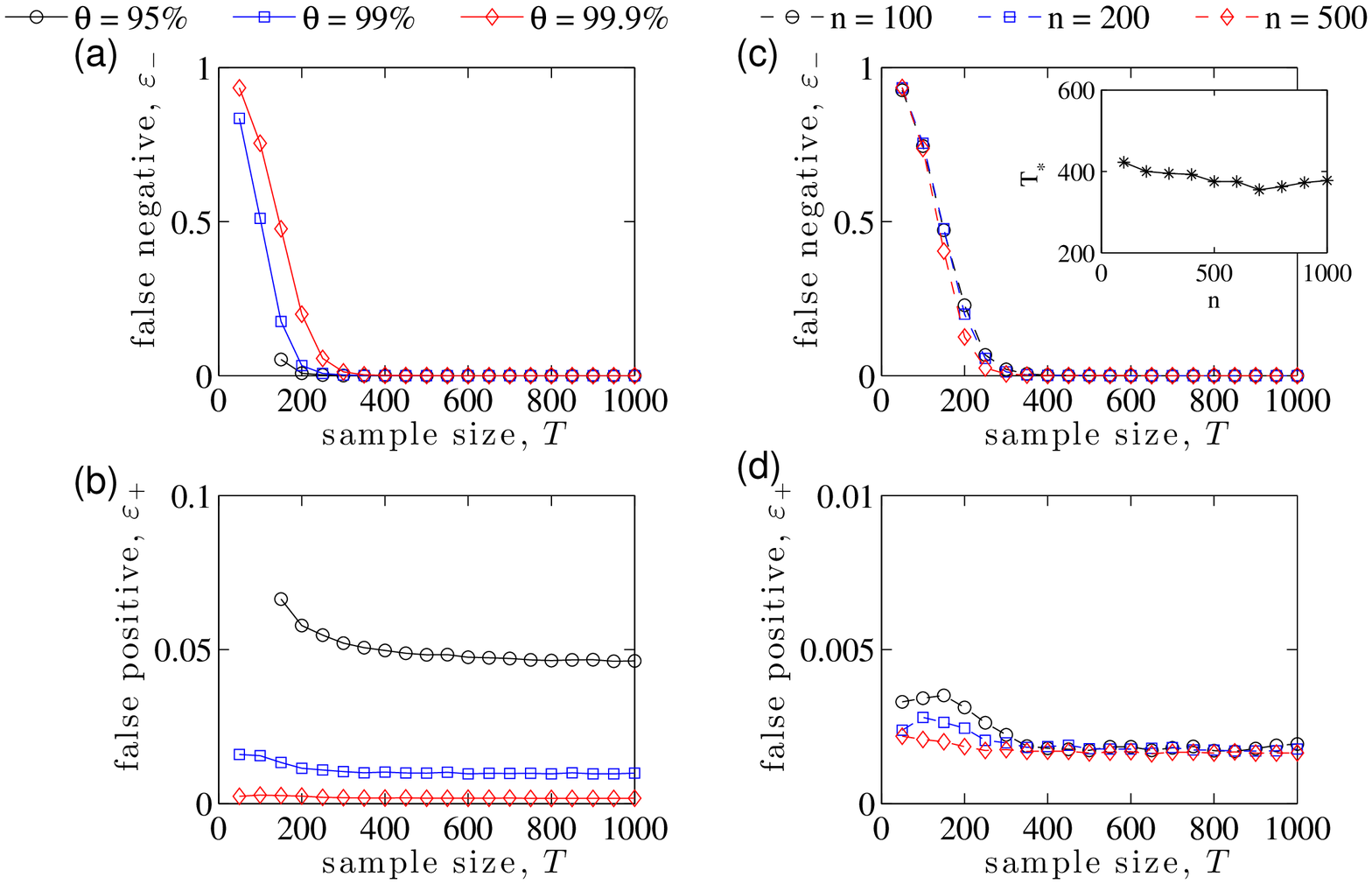}
\caption{Performance of the oCSE approach for causal network inference with different significance threshold for networks of various sizes.
The time series are generated from the Gaussian process defined in Eq.~\eqref{eq:main} using signed Erd\H{o}s-R\'enyi networks (see Sec 4.2 for details). False negative ratio (upper row) and false positive ratio (lower row) are defined in Eq.~\eqref{eq:inferror}.
(a-b) Inference error as a function of sample size $T$ for various significance levels $\theta$ used in the permutation test. Here networks have $n=200$ nodes with expected average degree $np=10$ and information diffusion rate $\rho(A)=0.8$.
(c-d) Inference error as a function of sample size $T$ for various network sizes. Here networks have the same expected average degree $np=10$ and information diffusion rate $\rho(A)=0.8$, and we use $r=1000$ permutations in the permutation test with $\theta=0.999$. Note that all three false negative curves in (c) appear to converge for $T\approx 300$. The critical sample size $T_*$ (defined as the minimum $T$ for which $\varepsilon_{-}<1-\theta$) as a function of the network size $n$ is shown in the inset of (c), suggesting the absence of scaling of $T_*$ in terms of $n$.
Each data point is obtained from averaging over $20$ independent simulations of the network dynamics, Eq.~\eqref{eq:main}.
}\label{newfig:8}
\end{figure}

\begin{figure}[htbp]
\centering
\includegraphics*[width=1\textwidth]{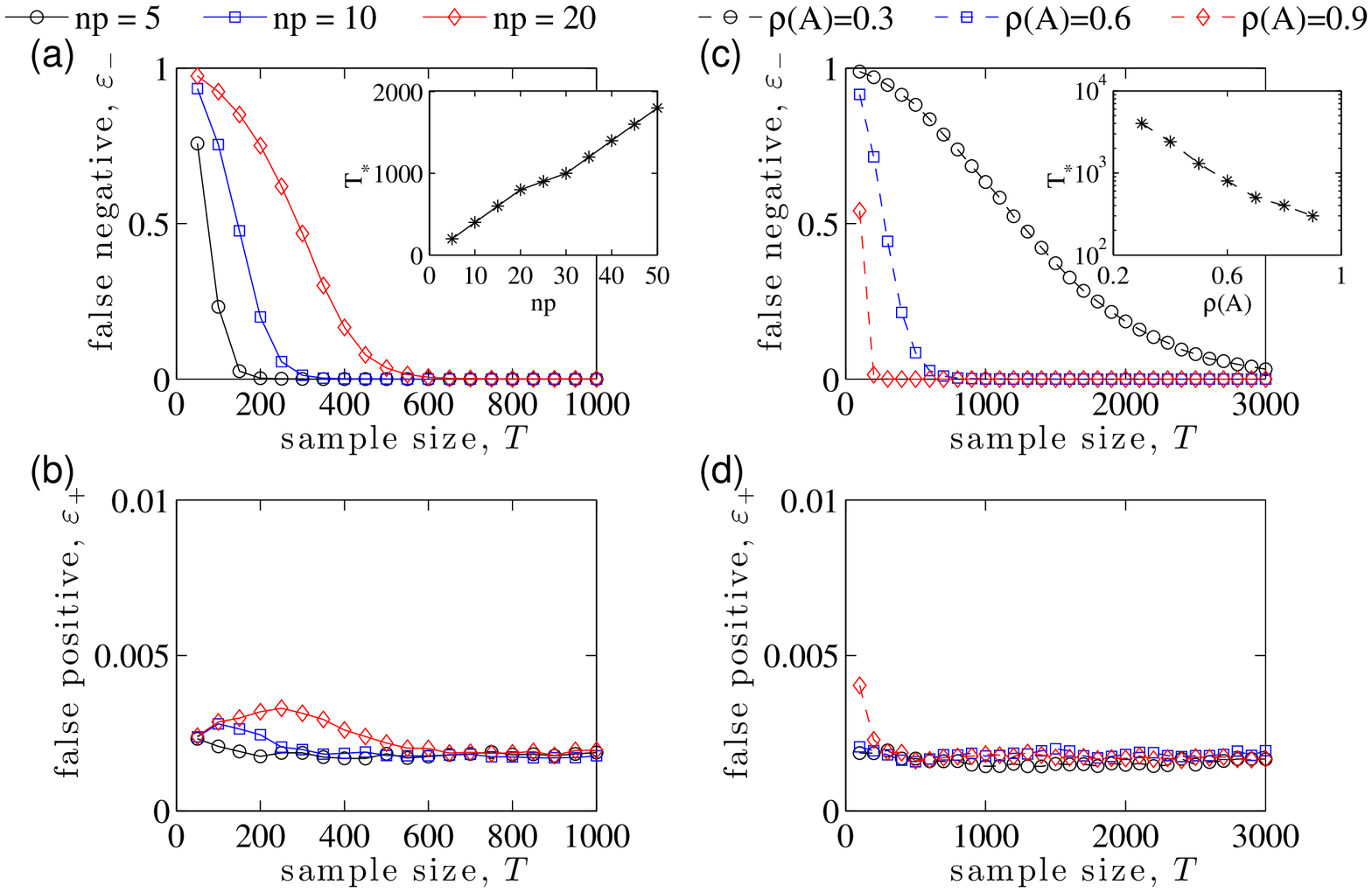}
\caption{
Performance of the oCSE approach for causal network inference for networks with different average degree and spectral radius.
The time series are generated from the Gaussian process defined in Eq.~\eqref{eq:main} using signed Erd\H{o}s-R\'enyi networks (see Sec 4.2 for details). False negative ratio (upper row) and false positive ratio (lower row) are defined in Eq.~\eqref{eq:inferror}.
(a-b) Inference error as a function of sample size for networks with various average degree $np$. Here the networks have the same size $n=200$ and spectral radius $\rho(A)=0.8$
The inset shows the critical sample size $T_*$ (see text) as a function of $np$.
(c-d) Inference error as a function of sample size for networks with various special radii $\rho(A)$. Here the networks have the same size $n=200$ and average degree $n=10$.
The permutation test used for the data in all panels involve $r=1000$ permutations with the significance threshold $\theta=0.999$.
Each data point is obtained from averaging over $20$ independent simulations of the network dynamics, Eq.~\eqref{eq:main}.
}\label{newfig:9}
\end{figure}

Next, we investigate the effect of sample size $T$ on the inference error for networks of different sizes. The results are shown in Fig.~\ref{newfig:8}(c-d). As expected, when $T$ increases, the false negative ratio decreases towards zero. Somewhat unexpectedly, the false positive ratio stays close to zero (in fact, close to the significance level $\theta$) even for relatively small sample size ($T$ as small as $50$ for networks of up to $500$ nodes). Furthermore, it appears that for networks of different sizes but the same average degree and information diffusion rate, the false negative ratios drop close to zero almost at the same sample size. To better quantify these effects, we define the {\it critical sample size} $T_*$ as the smallest number of samples for which the false negative ratio falls below $1-\theta$. As shown in the inset of Fig.~\ref{newfig:8}(c), for networks with the same average degree and information diffusion rate, the critical sample size $T_*$ remains mostly constant despite the increase of the network size. 
This result is unexpected. Traditionally, the network size $n$ represents a lower bound on sample size $T$ as any covariance matrix (e.g., application of the conditional Granger requires that $T>n$ for the invertibility of the covariance matrices). 
Our result surprisingly indicates that sample size $T$ does not need to scale with network size $n$ for accurate network inference, and highlights the fact that the oCSE approach is {\it scalable} and {\it data efficient}, with accuracy depending {\it not} on the size of the network, but rather on other network characteristics such as the density of links and spectral radius.

To strengthen our claim that for Erd\H{o}s-R\'enyi networks, performance of the causal inference by the oCSE approach depends on the density of links as measured by average degree and information diffusion rate as measured by the spectral radius rather than network size, we further investigate the dependence of inference error on these two additional parameters, $np$ and $\rho(A)$. As shown in Fig.~\ref{newfig:9}(a), for networks of the same size $n=200$ with fixed $\rho(A)=0.8$, the larger the average degree $np$, the larger the number of samples required  to reduce the false negative ratio to zero. In fact, as shown in the inset of Fig.~\ref{newfig:9}(a), the critical sample $T_*$ to reach $\varepsilon_{-}<1-\theta$ appears to scale {\it linearly} as a function of the average degree $np$, but not the network size (see the inset of Fig.~\ref{newfig:8}(c)). 
On the other hand, Fig.~\ref{newfig:9}(c-d) shows that the information diffusion rate, $\rho(A)$, seems to pose a harder constraint on accurate network inference: the smaller it is, the more samples that are needed for accuracy. In particular, as shown in the inset of Fig.~\ref{newfig:9}(c), the critical sample size appears to increase {\it exponentially} as $\rho(A)$ decreases towards zero.
Interestingly, as shown in Fig.~\ref{newfig:9}(b,d), the false positive ratios in both cases remain close to its saturation level around $1-\theta=10^{-3}$ even for very small sample size ($T\sim 50$), and this holds across networks with different average degree and different size (also see Fig.~\ref{newfig:8}(d)).

To briefly summarize these numerical experiments, we found that for the Gaussian process, practical causal network inference by the proposed oCSE overcomes fundamental limitations of previous approaches including conditional Granger and transfer entropy. One important advantage of the oCSE approach as suggested by the numerical results is that it often requires a relatively small number of samples to achieve high accuracy, making it a data-efficient method to use in practice. In fact, we found that for Erd\H{o}s-R\'enyi networks, the critical number of samples required for the false negatives to vanish does not depend on the network size, but rather depends on the density of links (as measured by average degree) and the information diffusion rate (as measured by the spectral radius of the network adjacency matrix). 
This is somewhat surprising because traditionally the network size poses as an absolute lower bound for the sample size in order for proper inversion of the covariance matrix (recent advances such as {\it Lasso} has partially resolved this issue by making specific assumptions of the model form and utilizing $l_1$ optimization techniques~\cite{Friedman2008,Tibshirani1996}).
On the other hand, our numerical results also suggest that only a very small number of samples is needed for the false positives to reach its saturation level. This level is inherently set by the significance threshold used in the permutation test rather than other network characteristics and can be systematically reduced by increasing the significance threshold and the number of permutations.

\section{Discussion and Conclusion}
Although time series analysis is broadly utilized for scientific research, the inference of large networks from relatively short times series data, and in particular causal networks describing ``cause-and-effect'' relationships, has largely remained unresolved.
The main contribution of this paper includes the theoretical development of causation entropy, an information-theoretic statistic designed for causality inference. 
Causation entropy can be regarded as a type of conditional mutual information which generalizes the traditional, unconditioned version of transfer entropy. When applied to Gaussian variables, causation entropy also generalizes Granger causality and conditional Granger causality.
We proved that for a general network stochastic process, the {causal parents} of a given node is exactly the minimal set of nodes that maximizes causation entropy, a key result which we refer to as the optimal causation entropy principle. Based on this principle, we introduced an algorithm for causal network inference called oCSE, which utilizes two algorithms to jointly infer the set of {causal parents} of each node. 

The effectiveness and data efficiency of the proposed oCSE approach were illustrated through numerical simulation of a Gaussian process on large-scale random networks. In particular, our numerical results show that the proposed oCSE approach consistently outperforms previous conditional Granger (with full conditioning) and transfer entropy approaches. Furthermore, inference accuracy using the oCSE approach generally requires fewer samples and fewer computations due to its aggregative nature: the conditioning set encountered in entropy estimation remains low-dimensional for sparse networks. The number of samples required for the desired accuracy does not appear to depend on network size, but rather, the density of links (or equivalently, the average degree of the nodes) and spectral radius (which measures the average rate at which information transfers through links). 
This makes oCSE a promising tool for the inference of networks, in particular large-scale sparse causal networks, as found in a wide range of real-world applications~\cite{BarratBook,Dorogovtsev2008,Newman2003,NewmanBook}.  Therefore we wish to emphasize that among all the details we presented herein, our oCSE-based algorithmic development (aggregative discovery jointly with progressive removal) is the most central contribution, serving as a method to systematically infer casual relationships from data generated by a complex interrelated process. In principle, we expect our two-step process given by Algorithms~\ref{alg:01} and~\ref{alg:02} to also be effective for network inference when the statistic is not necessarily causation entropy.

Several problems remain to be tackled. First, for general stochastic processes, exact expression of entropy is rarely obtainable.
Practical application of the oCSE therefore requires the development of non-parametric statistics for estimating causation entropy for general multi-dimensional random variables. An ideal estimation method should rely on as few assumptions about the form of the underlying variable as possible and be able to achieve the desired accuracy even for relatively small sample size. Several existing methods, including various binning techniques~\cite{Schindlera2007} and $k$-nearest neighbor estimates~\cite{Kraskov2004}, seem promising, but further exploration is necessary to examine their effectiveness~\cite{Hahs2011PRL}. Secondly, temporal stationarity assumptions are often violated in real-world applications. 
It is therefore of critical importance to divide the observed time series data into stationary segments~\cite{Wang2013}, allowing for the inference of causal networks that are {\it time-dependent}~\cite{Mantzaris2013}.
Finally, information causality suggests physical causality, but they are not necessarily equivalent~\cite{Hahs2011PRL,Pearl2009}. It is our goal to put this notion onto a more rigorous footing and further explore their relationships.

\section*{Acknowledgments}
We appreciate the insightful comments by C.~Cafaro, I.~Ipsen, J.~Skufca, G.~Song, and C. Tamon.
We thank Dr Samuel Stanton from the ARO Complex Dynamics and Systems Program for his ongoing and continuous support.

\appendix
\section{{Causal Inference of Finite-Order Markov Processes}}\label{app:A}
{
The main body of the paper deals with causal inference of a first-order stationary Markov process. Such framework can in fact be extended to any finite-order stationary Markov processes. The idea is to convert a finite-order process to a first-order one and define nodes in the causal network to be variables at different time layers.}

{Consider a stationary Markov process $\{Z_t\}$ of order $\tau$, which satisfies
\begin{equation}\label{eq:pZt}
	p(Z_{t}|Z_{t^-}) = p(Z_{t}|Z_{t-1},\dots,Z_{t-\tau})
\end{equation}
where $Z_{t^-}=[Z_{t-1},Z_{t-2},\dots]$ denotes the infinite past of $Z_t$.
Define a delay vector 
\begin{equation}
	X_t=[Z_{t},\dots,Z_{t-\tau+1}]. 
\end{equation}
Then, for every $x_t=[z_t,z_{t-1},\dots,z_{t-\tau+1}]$ and $x_{t^-}$, 
\begin{eqnarray}
	p(X_{t}=x|X_{t^-}=x_{t^-})&=&p(X_t=x_t|Z_{t-1}=z_{t-1},Z_{t-2}=z_{t-2},\dots)\nonumber\\
	&=&p(X_t=x_t|Z_{t-1}=z_{t-1},Z_{t-2}=z_{t-2},\dots,Z_{t-\tau}=z_{t-\tau})\nonumber\\
	&=&p(X_t=x_t|X_{t-1}=x_{t-1})
\end{eqnarray}
where the last step follows from Eq.~\eqref{eq:pZt} and the definition of $X_t$. 
See Fig.~\ref{fig:A1} for an example with $\tau=2$.
This shows that the process $\{X_t\}$ is indeed a first-order Markov process. The inference of the causal network is therefore converted into the identification of the causal parents of the nodes corresponding to $\{Z_t\}$ in the equivalent first-order process, for which the results in the main body of the paper apply so long as the conditions in Eq.~\eqref{eq:processconds} are met. }

\begin{figure}[htbp]
\centering
\includegraphics*[width=1\textwidth]{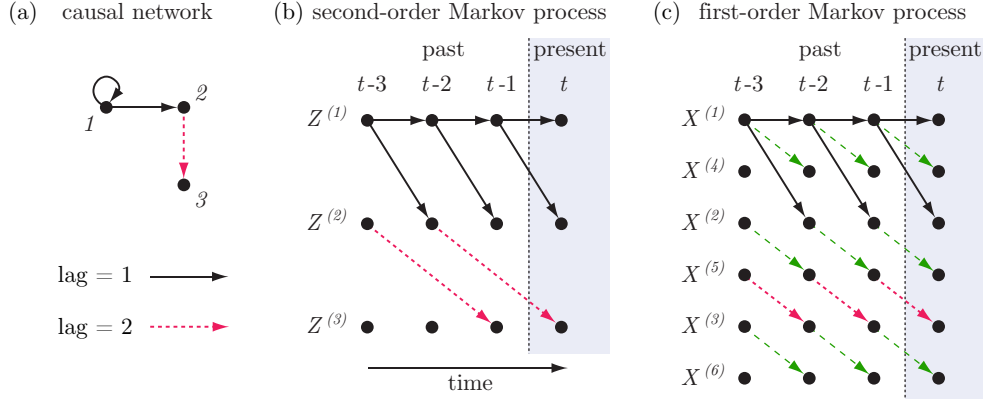}
\caption{
{
Converting a high-order Markov process into a first-order Markov process by making multiple instances of nodes. (a) A second-order Markov process on $n=3$ nodes, where causal relationships are across time lags of either 1 or 2 time time steps. We denote by $Z^{(i)}_t$ the state of node $i$ at time $t$.
(b) The flow of information for the second-order Markov process. Each row corresponds to a given node $i\in\{1,2,3\}$, and each column corresponds to the nodes' states $\{Z^{(i)}_t\}$ at a particular time $t$. Solid and dotted lines denote causal relationships across a time lag of 1 and 2 time steps, respectively.
(c) The flow of information for the equivalent first-order Markov process. Each row corresponds to a given node $i\in\{1,2,\dots,2n\}$, and each column corresponds to the nodes' states $\{X^{(i)}_t\}$ at a particular time $t$. For $i\in\{1,2,3\}$, the new variables $\{X^{(i)}_t\}$ are defined by $X^{(i)}_t = Z^{(i)}_t$ and $X^{(n+i)}_t = Z^{(i)}_{t-1}=X^{(i)}_{t-1}$. For Markov processes of order $\tau$, one can use the more general transformation $X^{((s-1)n+i)}_t = Z^{(i)}_{t-s+1}$ for nodes $i\in\{1,\dots,n\}$ and $s\in\{1,2,\dots,\tau\}$}.
}\label{fig:A1}
\end{figure}

{In practice, if the order of the underlying Markov process is {\it unknown}, then one needs to estimate it before being able to turn the process into a first-order one. The determination of Markov order has been a long-standing problem and is traditionally addressed by performing hypothesis tests based on computing a $\chi^2$ statistic~\cite{Bartlett1951}. The main disadvantage is that the $\chi^2$ distribution is only valid in the infinite-sample limit. A breakthrough was made recently by Pethel and Hahs~\cite{Pethel2014}, who developed a relatively efficient procedure for surrogate data generation which yields an exact test statistic valid for arbitrary sample size at the expense of increased computational burden.}

\section{{Necessity of the Faithfulness Assumption}}\label{app:B}
The faithfulness assumption is necessary for the ``true positive" statement in Theorem~\ref{thm:cse}(c) to be valid. As an example, consider a network of three nodes $X$, $Y$, and $Z$, and let 
\begin{equation}
	X_{t+1}=Y_t\mathbin{\oplus}Z_t, 
\end{equation}
where $\oplus$ denotes the ``exclusive or'' (xor) operation and $Y_t$ and $Z_t$ are Bernoulli random variables with probabilities
\begin{equation}
	P(Y_t=0)=P(Y_t=1)=P(Z_t=0)=P(Z_t=1)=0.5.
\end{equation}
It follows that
\begin{equation}
	C_{Y\rightarrow X}=C_{Z\rightarrow X}=0.
\end{equation}
However,
\begin{equation}
	C_{(Y,Z)\rightarrow X}=\log 2>0. 
\end{equation}
This results from the fact that multiple random variables can be mutually independent but not jointly independent. Expressed in terms of causal inference, it is possible that several variables jointly cause another variable, and this causal relationship cannot be decomposed. Such occurrences are believed to be rare and often explicitly excluded by making the faithfulness/stability assumption [48]. For example, in our above example it occurs only when all the discrete probabilities are exactly uniform, $p=0.5$, a situation that is unstable to perturbation. We exclude this situation from our study by imposing condition (3) in Eq.~\eqref{eq:processconds}.



\end{document}